\documentclass{siamltex}
\usepackage[utf8]{inputenc}
\usepackage{mathrsfs}
\usepackage{hyperref}
\usepackage{enumitem}
\usepackage{amsmath, amssymb}

\newcommand{\C}{\mathbb{C}}
\newcommand{\N}{\mathbb{N}}
\newcommand{\R}{\mathbb{R}}
\newcommand{\Q}{\mathbb{Q}}

\newcommand{\bigO}[1]{O\left(#1\right)}
\newcommand{\bfX}{\ensuremath{\mathbf{X}}}
\newcommand{\bfR}{\ensuremath{\mathbf{R}}}
\newcommand{\bfA}{\ensuremath{\mathbf{A}}}

\renewcommand{\emptyset}{\varnothing}
\newcommand{\fVC}{f_{|V}}
\newcommand{\fVCbfA}{f^\bfA_{|V^\bfA}}
\newcommand{\fV}{f_{|V\cap\R^n}}
\newcommand{\fVbfA}{f^\bfA_{|V^\bfA\cap\R^n}}

\newcommand{\fZbfA}{f^\bfA_{|Z^\bfA}}
\newcommand{\F}{\mathbf{F}}
\newcommand{\G}{\mathbf{G}}

\newcommand{\suitex}[1]{x^{(#1)}}
\DeclareMathOperator{\Singop}{\textnormal{\textsf{Sing}}}
\newcommand{\Sing}[1]{{\Singop\left(#1\right)}}
\newcommand{\fonction}[5]{\begin{aligned}#1: \\ ~\end{aligned}\begin{aligned} #2 & \longrightarrow #3 & \\ #4
    & \longmapsto #5 \end{aligned}}
\DeclareMathOperator{\Critop}{\textnormal{\textsf{Crit}}}
\newcommand{\Crit}[1]{\Critop\left(#1\right)}
\newcommand{\crit}{\Crit}
\newcommand{\critfV}{\Crit{f, V}}
\newcommand{\critfVbfA}{\Crit{f^\bfA, V^\bfA}}
\newcommand{\Roots}[1]{\mathrm{Roots}_\R\left(#1\right)}
\newcommand{\NEZOSJSCP}{\mathscr{O}_1}
\newcommand{\NEZOSISSAC}{\mathscr{O}_2}
\newcommand{\NEZOSC}{\mathcal{Q}}
\newcommand{\jac}[2]{\mathsf{Jac}\left(#1,#2\right)}
\newcommand{\jacc}[1]{\mathsf{Jac}\left(#1\right)}

\newcommand{\ProjLeft}[1]{\pi_{\leq #1}}
\newcommand{\Proj}[1]{\pi_{#1}}
\newcommand{\ProjISSAC}[1]{\varphi_{#1}}
\newcommand{\ZariskiClosure}[1]{\overline{#1}^\mathcal{Z}}
\newcommand{\Variety}[1]{V\left(#1\right)}
\newcommand{\PolRing}{\mathbb{Q}\left[\bfX\right]}
\newcommand{\PolRingC}{\mathbb{C}\left[\bfX\right]}

\newcommand{\IdealAngle}[1]{\left\langle #1 \right\rangle}
\newcommand{\nV}{s}
\newcommand{\GLnQ}{\mathrm{GL}_n(\Q)}
\newcommand{\GLnC}{\mathrm{GL}_n(\C)}
\newcommand{\PropReg}[1]{\mathfrak{R}\left(#1\right)}
\newcommand{\PropPJSC}[1]{\mathfrak{P}_1\left(#1\right)}
\newcommand{\PropPISSAC}[1]{\mathfrak{P}_2\left(#1\right)}
\newcommand{\PropOpt}[1]{\textnormal{Opt}\left(#1\right)}

\newcommand{\NEZOS}{\mathscr{O}}
\newcommand{\SetContaining}{\textnormal{\textsf{SetContainingLocalExtrema}}}
\newcommand{\Optimize}{\textnormal{\textsf{Optimize}}}
\newcommand{\FindInfimum}{\textnormal{\textsf{FindInfimum}}}
\newcommand{\Par}[1]{\left(#1\right)}
\newcommand{\ClosedN}{\mathcal{U}}
\newcommand{\locextr}{\ell}
\newcommand{\eps}{\varepsilon}
\newcommand{\ClosedBall}[1]{\overline{B}\left(#1\right)}

\newcommand{\assumptR}{\bfR}
\newcommand{\listS}{\mathsf{ListSamplePoints}}
\newcommand{\listP}{\mathsf{ListCriticalPoints}}
\newcommand{\listC}{\mathsf{L}_\mathsf{sat}}
\newcommand{\polNP}{P_\mathsf{NP}}
\newcommand{\SamplePoints}{\textnormal{\textsf{RealSamplePoints}}}
\newcommand{\SetOfNonProperness}{\textnormal{\textsf{Set\-Of\-Non\-Properness}}}
\newcommand{\RealRootIsolation}{\textnormal{\textsf{RealRootIsolation}}}
\newcommand{\IsEmpty}{\textnormal{\textsf{IsEmpty}}}
\newcommand{\Rpuiseux}{\R\left\langle\eps\right\rangle}

\newcommand{\limzero}{\lim_0}
\newcommand{\ext}[1]{\mathsf{ext}\left(#1\right)}

\newcommand{\CCbfA}{C^\bfA}
\newcommand{\LinSpace}[1]{\Variety{\bfX_{\leq #1}}}
\newcommand{\NP}[2]{{\mathsf{NP}\left(#1,#2\right)}}
\newcommand{\VVA}[3]{{\mathscr{C}\left(#1,#2,#3\right)}}
\newcommand{\VPC}[3]{{\mathscr{P}\left(#1,#2,#3\right)}}
\newcommand{\VVAtot}[2]{{\mathscr{C}\left(#1,#2\right)}}
\newcommand{\VPCtot}[2]{{\mathscr{P}\left(#1,#2\right)}}
\newcommand{\VSP}[1]{\mathscr{S}\left(#1\right)}

\let\oldinf\inf
\renewcommand{\inf}{\displaystyle\oldinf}

\newtheorem{remark}[theorem]{Remark}

\newtheorem{ex}{Example}

\newenvironment{example}[1]
{\begin{ex}[#1] \normalfont}
{\end{ex}}

\title{Probabilistic Algorithm for Polynomial Optimization over a Real Algebraic Set}

\begin{document}
\renewcommand{\thefootnote}{\fnsymbol{footnote}}
\author{Aur\'elien Greuet \footnotemark[1] \ \footnotemark[2] \ \footnotemark[3] \ \footnotemark[4]
  \and Mohab Safey El Din \footnotemark[2] \ \footnotemark[4]}

\footnotetext[1]{{Laboratoire de Math\'ematiques (LMV-UMR8100)

    Universit\'e de Versailles-Saint-Quentin

    45 avenue des \'Etats-unis, 78035 Versailles Cedex, France
}}
\footnotetext[2]{Sorbonne Universit\'es, Univ. Pierre et Marie Curie (Paris 6)

  INRIA, Paris Rocquencourt Center, POLSYS Project, 

  LIP6/CNRS, UMR 7606, 

  Institut Universitaire de France, 

  {Mohab.Safey@lip6.fr} }

\footnotetext[1]{{Universit\'e de Lille 1
    
    Cité scientifique - bâtiment M3
    
    59655 Villeneuve d'Ascq, France

    {Aurelien.Greuet@univ-lille1.fr}}}

\footnotetext[4]{Mohab Safey El Din and Aur\'elien Greuet are
  supported by the GEOLMI grant (ANR 2011 BS03 011 06) of the French
  National Research Agency.}

\renewcommand{\thefootnote}{\arabic{footnote}}

\maketitle
\begin{abstract}
  Let $f, f_1, \ldots, f_\nV$ be $n$-variate polynomials with rational
  coefficients of maximum degree $D$ and let $V$ be the set of common
  complex solutions of $\F=(f_1,\ldots, f_\nV)$. We give an algorithm
  which, up to some regularity assumptions on $\F$, computes an {\em
    exact} representation of the global infimum $f^\star$ of the
  restriction of the map $x\to f\Par{x}$ to ${V\cap\R^n}$, i.e. a
  univariate polynomial vanishing at $f^\star$ and an isolating
  interval for $f^\star$. Furthermore, it decides whether $f^\star$ is
  reached and if so, it returns $x^\star\in V\cap\R^n$ such that
  $f\Par{x^\star}=f^\star$.

  This algorithm is {\em probabilistic}. It makes use of the notion of
  polar varieties. Its complexity is essentially {\em cubic} in
  $\Par{\nV D}^n$ and linear in the complexity of evaluating the
  input. This fits within the best known {\em deterministic}
  complexity class $D^{O(n)}$.

  We report on some practical experiments of a first implementation
  that is available as a {\sc Maple} package. It appears that it can
  tackle global optimization problems that were unreachable by
  previous exact algorithms and can manage instances that are hard to
  solve with purely numeric techniques. As far as we know, even under
  the extra genericity assumptions on the input, it is the first
  probabilistic algorithm that combines practical efficiency with good
  control of complexity for this problem.
\end{abstract}



\section{Introduction}

Let $\bfX=X_1, \ldots, X_n$ be indeterminates, $f, f_1, \ldots, f_\nV$
be polynomials in $\PolRing$ of maximal degree $D$ and $V=V(\F)$ be
the set of common complex solutions of $\F=(f_1,\ldots, f_\nV)$. We
focus on the design and the implementation of {\em exact} algorithms
for solving the polynomial optimization problem which consists in
computing and exact representation of the global infimum
$f^\star=\inf_{x\in V\cap\R^n} f\left(x\right)$.  It is worth to note
that, at least under some genericity assumptions, polynomial
optimization problems whose constraints are non-strict inequalities
can be reduced to the one with polynomial equations (see
e.g. \cite{BGHS14} and references therein).


\paragraph*{Motivation and prior work}
While polynomial optimization is well-known to be {\cal NP}-hard (see
e.g. \cite{nesterov00}), it has attracted a lot of attention since it
appears in various areas of engineering sciences (e.g. control theory
\cite{PositivePolControl,HSK03}, static analysis of programs
\cite{Cousot05,Monniaux10}, computer vision \cite{AhAgTh12, AhStTh11},
economics, etc.). In this area, one challenge is to
combine practical efficiency with reliability for polynomial
optimization solvers.

One way to reach this goal is to relax the polynomial optimization
problem by computing algebraic certificates of positivity proving
lower bounds on $f^\star$. This is achieved with methods computing
sums of squares decompositions of polynomials. In this context, one
difficulty is to overcome the fact that a nonnegative polynomial is
not necessarily a sum of squares. Various techniques have been
studied, see e.g. \cite{DeNiPo07, GrGuSaZh, Vui-constraints,
  GloptiPoly, lasserre01globaloptimization, NiDeSt,
  schweighofer}. These approaches use semi-definite programming
relaxations (\cite{parriloPHD, shor}) and numerical solvers of
semi-definite programs.  Sometimes, a sum of squares decomposition
with rational coefficients can be recovered from such a decomposition
computed with floating point coefficients (see \cite{KLYZ12, PaPe}).
Algorithms for computing sums of squares decompositions with rational
coefficients have also been designed \cite{QSZ13, SZ10}. Some cases of
ill-conditionedness have been identified (\cite{gcv_sos}), but there
is no general method to overcome them. It should also be noticed that
techniques introduced to overcome situations where a non-negative
polynomial is not a sum of squares rely on using gradient varieties
\cite{DeNiPo07, GrGuSaZh, NiDeSt} which are close to polar varieties
introduced in the context of symbolic computation for studying real
algebraic sets (see e.g. \cite{BaGiHeMb97, BaGiHeMb01, BGHSS, SaSc}),
quantifier elimination (see e.g. \cite{HS09, HS12}) or connectivity
queries (see e.g. \cite{SaSc11, din2013nearly}).


Another way to combine reliability and practical efficiency is to
design algorithms relying on symbolic computation that solve the
polynomial optimization problem. Indeed, it can be seen as a special
quantifier elimination problem over the reals and a goal would be to
design a dedicated algorithm whose complexity meets the best known
bounds and whose practical behaviour reflects its complexity.

Quantifier elimination over the reals can be solved by the cylindrical
algebraic decomposition algorithm \cite{Collins75}. This algorithm
deals with general instances and has been studied and improved in many
ways (see e.g. \cite{Brown99, Collins98, CollinsHong, Hong92,
  McCallum98}). However, its complexity is doubly exponential in the
number of variables. In practice, its best implementations are limited
to non trivial problems involving $4$ variables at most.

In \cite{BPR96}, a deterministic algorithm whose complexity is singly
exponential in the number of alternations of quantifiers is given. On
polynomial optimization problems, this specializes to an algorithm for
polynomial optimization that runs in time $D^{O(n)}$ (see
\cite[Chapter 14]{BaPoRo06}). The techniques used to get such
complexity results such as infinitesimal deformations did not provide
yet practical results that reflect this complexity gain. While in some
special cases, practical algorithms for one-block quantifier
elimination problems have been derived by avoiding the use of
infinitesimals \cite{HS09, HS12}, the problem of obtaining fast
algorithms in theory and in practice for polynomial optimization
remained open.

Thus, our goal is to obtain an efficient algorithm for solving the
polynomial optimization problem in theory and in practice. Thus, we
expect its complexity to lie within $D^{O(n)}$ operations but with a
good control on the complexity constant in the exponent. We allow to
have regularity assumptions on the input that are reasonable in
practice (e.g. rank conditions on the Jacobian matrix of the input
equality constraints). We also allow probabilistic algorithms provided
that probabilistic aspects do not depend on the input but on random
choices performed when running the algorithm.



A first attempt towards this goal is in \cite{Safey08}. Given a
$n$-variate polynomial $f$ of degree $D$, a probabilistic algorithm
computing $\displaystyle\inf_{x\in\R^n}f\Par{x}$ in $\bigO{n^7D^{4n}}$
operations in $\Q$ is given. Moreover, it is practically efficient
and has solved problems intractable before (up to $6$ variables). Our
goal is to generalize this approach to the case of equality
constraints and get an algorithm whose complexity is essentially cubic
in $\Par{\nV D}^n$ and linear in the evaluation complexity of the input.

\paragraph*{Main results} We provide a probabilistic algorithm based
on symbolic computation solving the polynomial optimization problem up
to some regularity assumptions on the equality constraints whose
complexity is essentially cubic in $\Par{\nV D}^n$. We also provide an
implementation of it and report on its practical behaviour which
reflects its complexity and allows to solve problems that are either
hard from the numerical point of view or unreachable by previous
algorithms based on symbolic computation.

Before describing these contributions in detail, we start by stating
our regularity assumptions which hold on the equality constraints. In
most of applications, the Jacobian matrix of $\F=(f_1, \ldots, f_\nV)$
has maximal rank at all points of the set of common solutions of
$\F$. In algebraic terms, this implies that this solution set is
smooth of co-dimension $\nV$, complete intersection and the ideal
generated by $\F$ (i.e. the set of algebraic relations generated by
$\F$) is radical.

Our regularity assumptions are a bit more general than the situation
we just described. In the sequel, we say that $\F$ satisfies 
assumption $\assumptR$ if the following holds:
\begin{itemize}
\item the ideal $\IdealAngle{\F}$ is radical,
\item $\Variety{\F}$ is equidimensional of dimension $d>0$,
\item $\Variety{\F}$ has finitely many singular points.
\end{itemize}

Under these assumptions, we provide an algorithm that decides the
existence of $f^\star= \inf_{x\in\Variety{\F}\cap\R^n}f\Par{x}$ and,
whenever $f^\star$ exists, it computes an exact representation of it (i.e. a
univariate polynomial vanishing at $f^\star$ and an isolating interval
for $f^\star$).  It can also decide if $f^\star$ is reached and if
this is the case it can return a minimizer $x^\star$ such that
$f\Par{x^\star}=f^\star$. 

We count arithmetic operations $+, -, \times, \div$ in $\Q$ and sign
evaluation at unit cost. We use the soft-O notation:
$\widetilde{O}(a)$ indicates the omission of polylogarithmic factors
in $a$.  The complexity of the algorithm described in this paper is
essentially cubic in $\Par{\nV D}^n$ and linear in the complexity of
evaluating $f$ and $\F$. For instance if the Jacobian matrix of $\F$
has full rank at all points of $\Variety{\F}$ (this is a bit more
restrictive than $\assumptR$) then the algorithm performs
\[\widetilde{O}\left(
    L\left(\sqrt[3]{2} \left(s+1\right)D\right)^{3(n+2)}\right)\]
arithmetic operations in $\Q$ (see Theorem \ref{thm:complexity} for
the general case in Section \ref{sec:complexity}).

Note that this algorithm is a strict generalization of the one given
in \cite{Safey08}. Note also that when the infimum is reached, we
compute a minimizer without any assumption on the dimension of the set
of minimizers.

Our algorithm follows a classical pattern which is used for quantifier
elimination over the reals. It first performs a change of coordinates
to ensure some technical assumptions that are satisfied in general
position. Then, roughly speaking, it computes an ordered finite set of
real numbers containing $f^\star$. Moreover, for any interval between
two consecutive numbers in this set is either contained in
$f\left(\Variety{\F}\cap\R^n\right)$ or has an empty intersection with
$f\left(\Variety{\F}\cap\R^n\right)$.

To compute this set, we use geometric objects which are close to the
notion of polar varieties which, under $\assumptR$, are critical loci
of some projections ; we refer to \cite{BGHSS} for more details on 
several properties of polar varieties and to \cite{BGHS14} for
geometric objects similar to the ones we manipulate in a more
restrictive context. Our modified polar varieties are defined
incrementally and have a degree which is well controlled (essentially
singly exponential in $n$). Algebraic representations of these
modified polar varieties can be computed using many algebraic
algorithms for polynomial system solving. Properties of the systems
defining these modified polar varieties are exploited by some
probabilistic algebraic elimination algorithms (see e.g. the geometric
resolution algorithm \cite{GiLeSa01} and references therein) which
allows to state our complexity results.

Our implementation is based on Gr\"obner basis computations which have
a good behaviour in practice (see also \cite{FGS13, Sp13} for
preliminary complexity estimates explaining this behaviour). Recall
that most of algorithms for computing Gr\"obner bases are
deterministic.  This implementation is available at
\url{http://www-polsys.lip6.fr/~greuet/}. We describe it in Section
\ref{sec:practical}; in particular, we show how to check if the
generic assumptions required for the correctness are satisfied after
performing a linear change of coordinates.  We report on experiments
showing that its practical performances outperform other
implementations of previous algorithms using symbolic computation and
can handle non-trivial problems which are difficult from the numerical
point of view.


\paragraph*{Plan of the paper}
We introduce notations and definitions of geometric objects in Section
\ref{sec:notations}. Section \ref{sec:algo} describes the algorithm
and its subroutines. Section \ref{sec:correctness} is devoted to the
proof of correctness of the algorithm, under some geometric
assumptions on some geometric objects depending on the input. Then in
Section \ref{sec:generic}, we prove that these assumptions are true up
to a generic change of coordinates on the input. Finally, the
complexity is analyzed in Section \ref{sec:complexity}. Some details
on the implementation and practical results are presented in Section
\ref{sec:practical}.



\section{Notations and Basic Definitions}\label{sec:notations}

This section introduces basic geometric objects that our algorithm
manipulates. We also make explicit the regularity assumptions that are
needed to ensure correctness of the algorithm. It is probabilistic
because it requires some generic linear change of variables that are
necessary to ensure some properties that are made explicit too. 

\subsection{Standard notions}
We start by defining basic objects we consider in the sequel. Most of
the notions presented below are described in detail in
\cite{Shafarevich77}. This culminates with the notion of singular and
critical points of a polynomial map. In our context, these notions are
important since the polynomial map under consideration reaches its
extrema at these points.

\paragraph{Algebraic sets}
Let $\bfX=\Par{X_1,\dots,X_n}$ and
$\F=\left\{f_1,\dots,f_\nV\right\}\subset\PolRing$. The algebraic
variety $\Variety{\F}$ is the set $\left\{x\in\C^n\mid
  f_1\Par{x}=\cdots=f_\nV\Par{x}=0\right\}$. The \emph{Zariski
  topology} on $\C^n$ is a topology where the closed sets are the
algebraic varieties. Given a set $U\subset\C^n$, the Zariski closure
of $U$, denoted by $\ZariskiClosure{U}$, is the closure of $U$ for the
Zariski topology. It is the smallest algebraic variety containing $U$.
A Zariski open set is the complement of a Zariski closed set. An
algebraic variety $V$ is $\Q$-\emph{reducible} if it can be written as
the union of two proper algebraic varieties defined by polynomials
with coefficients in $\Q$, \emph{irreducible} else. In this paper, all
the algebraic sets we will consider will be defined by polynomials
with coefficients in $\Q$; thus the notion of reducibility will refer
to $\Q$-reducibility. 

For any variety $V$, there exist irreducible varieties
$V_1,\dots,V_\ell$ such that for $i\neq j$, $V_i\not\subset V_j$ and
such that $V = V_1\cup\dots\cup V_s$. The algebraic varieties $V_i$
(for $1\leq i \leq \ell$) are the irreducible components of $V$. The
decomposition of $V$ as the union of its irreducible components is
unique. 
The dimension of $V=\Variety{f_1,\dots,f_s}$ is the Krull dimension of
its coordinate ring, that is the maximal length of the chains
$p_0\subset p_1\subset\dots\subset p_d$ of prime ideals of the
quotient ring $\PolRingC/\IdealAngle{f_1,\dots,f_s}$ (see
\cite[Chapter 8]{eisenbud}). We write $\dim V = d$. The variety is
\emph{equidimensional} of dimension $d$ if and only if its irreducible
components have dimension $d$.

\paragraph{Polynomial mapping and Jacobian matrices}
Given $f\in\PolRing$, by abuse of notation, we write $f$ for the
polynomial mapping $x\mapsto f\Par{x}$. Given
$\F=\left\{f_1,\dots,f_\nV\right\}\subset\PolRing$, $\jacc{\F}$ is the
Jacobian matrix $\displaystyle\left(\frac{\partial f_i}{\partial
    X_j}\right )_{\begin{subarray}{c}1\leq i\leq \nV\\1\leq j \leq n
  \end{subarray}}$. Likewise, $\jac{\F}{k}$ denotes the truncated
Jacobian matrix of size $p\times\left(n-k+1\right)$ with respect to
the variables $X_k,\dots,X_n$.

\paragraph{Projections}
Let $f\in\PolRing$ and $T$ be a new indeterminate. For $1\leq i\leq
n$, $\ProjLeft{i}$ is the projection 
\[ 
\begin{aligned}
  \Variety{f-T}\cap \left(V\times\C\right) & \longrightarrow \C^{i+1}\\
  (x_1, \ldots, x_n,t) & \longmapsto \left(x_1,\dots,x_i,t\right).
\end{aligned}.\] For $i=0$, the projection
$\ProjLeft{0}\colon\Par{x_1,\dots,x_n, t}\longmapsto t$ is denoted by
$\Proj{T}$.

Given a set $W\subset\C^{n}$, the set of non-properness of the
restriction of $\Proj{T}$ to
$\left(W\times\C\right)\cap\Variety{f\!-\!T}$ is denoted by
$\NP{\Proj{T}}{W}$. This is the set of values $t\in\C$ such that for
all closed neighbourhoods $\mathcal{O}$ of $t$ (for the euclidean
topology), $\Proj{T}^{-1}\Par{\mathcal{O}}\cap
\left(W\times\C\right)\cap\Variety{f-T}$ is not closed and bounded.

\paragraph{Change of coordinates}
Given $\bfA\in\GLnQ$, $f^\bfA$ (resp. $\F^\bfA$, $V^\bfA$) is the the
polynomial $f\left(\bfA\bfX^{}\right)$ (resp. the family
$\left\{f_1^\bfA,\dots,f_\nV^\bfA\right\}$, the variety
$\Variety{\F^\bfA}$). We also denote by $f^\bfA$ the polynomial
mapping $x\mapsto f^\bfA\Par{x}$. A property on an algebraic set
$\Variety{g_1,\dots,g_p}$ is called generic if there exists a
non-empty Zariski open subset of $\GLnC$ such that for all matrices
$\bfA\in\GLnQ$ in this open set, the property holds for
$\Variety{g^\bfA_1,\dots,g^\bfA_p}$.

\paragraph{Regular and singular points}
The Zariski tangent space to $V$ at $x\in V$ is the vector space $T_x
V$ defined by the equations
  \[\frac {\partial f}{\partial X_1}(x) v_1
  + \cdots + \frac {\partial f}{\partial X_n}(x) v_n=0,\] for all
  polynomials $f$ that vanish on $V$. If $V$ is equidimensional, the
  \emph{regular points} on $V$ are the points $x\in V$ where $\dim(T_x
  V)=\dim(V)$; the \emph{singular points} are all other points. The
  set of singular points of $V$ is denoted by $\Sing{V}$. If
  $V=\Variety{\F}$ is equidimensional of dimension $d$ then the set of
  singular points is the set of points in $V$ where the minors of size
  $n-d$ of $\jacc{\F}$ vanish.

\paragraph{Critical points}
Assume that $V=\Variety{\F}$ is equidimensional of dimension $d$. A
point $x\in V\setminus\Sing{V}$ is a critical point of $\fVC$, the
restriction of $f$ to $V$, if it lies in the variety defined by all
the minors of size $n-d+1$ of $\jacc{[f,\F]}$.

We denote by $\Crit{f,V}$ the algebraic variety defined as the
vanishing set of
\begin{itemize}
\item the polynomials in $\F$,
\item and the minors of size $n-d+1$ of $\jacc{[f,\F]}$.
\end{itemize}

\subsection{Definitions}
Our algorithm works under some regularity assumptions that are
reasonable from the application viewpoint. These assumptions are
based on some rank conditions of the Jacobian matrix of the input
constraints $\F$. These rank conditions are sufficient to be able to
characterize from $\F$ the critical points of the restriction of the
map $x \to f(x)$ to $V(\F)$.

We start by defining these regularity assumptions and
next, we introduce basic geometric objects (modified polar varieties)
that are built upon the notions of singular and critical points and
that we use further to construct objects of dimension at most $1$ on
which we can ``read'' the global infimum of a polynomial map.

\paragraph{Assumptions of regularity}

Let $\F\subset\PolRing$ be a polynomial family such that
$\IdealAngle{\F}$ is radical and $V=\Variety{\F}$ is equidimensional
of dimension $d$. In this context, the set of singular points of $V$
is the variety $\Sing{V}$ defined as the vanishing set of
\begin{itemize}
\item the polynomials in $\F$,
\item and the minors of size $n-d$ of $\jacc{\F}$,
\end{itemize}
The algebraic variety $V$ is \emph{smooth} if $\Sing{V}=\emptyset$.

It is well known that local extrema are critical or singular
points. Thus, it is natural to compute them. In order to be able to
compute these points, we will assume some properties of regularity on
the inputs.

The polynomial family $\F$ satisfies assumptions $\assumptR$ if
\begin{itemize}
\item the ideal $\IdealAngle{\F}$ is radical,
\item $\Variety{\F}$ is equidimensional of dimension $d>0$,
\item $\Variety{\F}$ has finitely many singular points.
\end{itemize}

Remark that if $V$ satisfies assumptions $\assumptR$ then the variety
$\Crit{f,V}$ defined above as the zero-set of minors of the Jacobian
matrix of the system is the union of the critical points of $\fVC$ and
$\Sing{V}$. Hence, it contains all the points at which the local
extrema are reached.

In this paper, we consider a polynomial family
$\F=\left\{f_1,\dots,f_\nV\right\}$ that satisfies assumptions
$\assumptR$. We denote by $V$ the algebraic variety $\Variety{\F}$.

\paragraph{Sample points and modified polar varieties }
From a computational point of view, the characterization of critical
and singular points as solutions of a polynomial system is not
sufficient. When they are in finite number, we will compute
parametrizations of these sets.

However, an infimum is not necessarily reached. It can be an
asymptotic value, that is the limit of a sequence
$f\left(x_\ell\right)$, where $\left(x_\ell\right)_{\ell\in\N}\subset
V\cap\R^n$ tends to infinity.

Our goal is then to construct geometric objects that can be used to
compute a parametrization of some critical points and asymptotic
values. This is the motivation of the following definition.

\begin{definition}\label{def:geomobjects}
  For $1\leq i\leq d-1$, let $\VVA{f}{\F}{i}$ be the algebraic variety
  defined as the vanishing set of 
  \begin{itemize}
  \item the polynomials in $\F$,
  \item the minors of size $n-d+1$ of $\jac{\left[f,\F\right]}{i+1}$,
  \item and the variables $X_1,\dots,X_{i-1}$.
  \end{itemize}
  By convention, $\VVA{f}{\F}{d} = V\cap\Variety{X_1,\dots,X_{d-1}}$. Let
  \[\VVAtot{f}{\F} = \bigcup_{1\leq i\leq d} \VVA{f}{\F}{i}.\]

  For $1\leq i\leq d-1$, let
  $\VPC{f}{\F}{i}=\ZariskiClosure{\VVA{f}{\F}{i}\setminus\critfV}\cap\critfV$. For
  $i=d$, let $\VPC{f}{\F}{d} = \VVA{f}{\F}{d}$. Finally, let
  \[\VPCtot{f}{\F} = \bigcup_{1\leq i\leq d} \VPC{f}{\F}{i}.\]
\end{definition}
We will prove that up to a generic linear change of coordinates
$\ZariskiClosure{\VVA{f}{\F}{i}\setminus\critfV}$
(resp. $\VPC{f}{\F}{i}$) has dimension at most $1$ (resp. $0$).

Remark that under assumptions $\assumptR$, $\VVAtot{f}{\F}$ is the
union of
\begin{itemize}
\item the set of singular points $\Sing{V}$,
\item the intersection of $\Variety{X_1,\dots,X_{i}}$ and the critical
  locus of the projection $\ProjLeft{i}$ restricted to
  $\left(V\times\C\right)\cap\Variety{f-T}$, for $1\leq i\leq d$.
\end{itemize}

This definition is inspired by the one of the polar varieties (see
\cite{BaGiHeMb97, BaGiHeMb01,BGHSS,SaSc, SaSc04}). 

We denote by $\VSP{\F}$ any finite set that contains at least one
point in each connected component of $V\cap\R^n$. Such a set can be
efficiently computed using e.g. \cite{SaSc}.

\subsection{Some useful properties}\label{sec:propertiesoptimization}
As already mentioned, we will prove that up to a generic change of
coordinates, $\VVAtot{f}{\F}\setminus \Crit{f,V}$ has dimension at
most one. {F}rom this, we will deduce that the set of asymptotic
values of $f$ is the set of non-properness of the restriction to
$\Variety{f-T}\cap\left(\VVAtot{f}{\F}\times\C\right)$ of the
projection $\pi_T$. Since $\VVAtot{f}{\F}\setminus \Crit{f,V}$ has
dimension at most one, this set of non-properness is finite and can be
computed algorithmically \cite{LaRo07, SaSc04}.

Moreover, we will prove that up to a generic change of coordinates,
$\VPCtot{f}{\F}$ can be used to compute a finite set of points whose
image by $f$ contains all the reached local extrema.

In order to identify the global infimum among these values of
non-properness and the reached local extrema, some properties are
needed. For simplicity, these properties are summarized in the
following definition.

\begin{definition}\label{def:propopt}
  Let $V$ be an algebraic set and $W$ be a subset of $\R$, we say that
  property $\PropOpt{W, V}$ holds if:
  \begin{enumerate}
  \item $W$ is finite,\label{item:prooptfinite}
  \item $W$ contains every local extremum of $\fV$,\label{item:prooptextr}
  \item let $W=\left\{a_1,\dots,a_k\right\}$, $a_0=-\infty$ and
    $a_{k+1}=+\infty$. There exists a non-empty Zariski open set
    $\NEZOSC\subset\C$ such that for all $0\leq i \leq k$ and all
    couples $(t, t')$ in $\left]a_i,a_{i+1}\right[$
\[
f^{-1}\left( t \right)\cap V\cap\R^n=\emptyset\Longleftrightarrow f^{-1}\left( t' \right)\cap V\cap\R^n=\emptyset. 
\]
  \end{enumerate}
\end{definition}

Now, we define the set $W(f,\F)$ (or simply $W$ when there is no
ambiguity on $\F$) as the union of $f\left(\VSP{\F}\right)$,
$f\left(\VPCtot{f}{\F}\right)$ and the set of non-properness of the
restriction of the projection $\pi_T$ to
$\Variety{f-T}\cap\left(\VVAtot{f}{\F}\times\C\right)$.  Assuming that
$\F$ satisfies $\assumptR$, our goal is to prove that
$\PropOpt{W(f,\F), V}$ is satisfied, up to a generic change of
coordinates.

\subsection{Genericity properties}\label{sec:genericnotations}
In order to do this, we will use some geometric properties that are
true up to a generic linear change of coordinates. We define these
properties in the next paragraph. Assuming these generic properties,
we prove that $\PropOpt{W(f, \F), V}$ holds in Section
\ref{sec:correctness}.

A value $c\in\R$ is isolated in $f\Par{V\cap\R^n}$ if and only if
there exists a neighborhood $\mathcal{B}$ of $c$ such that
$\mathcal{B}\cap f\Par{V\cap\R^n}=\left\{c\right\}$. Given
$f\in\PolRing$ and $\F\subset\PolRing$, we consider the following
properties.
\begin{itemize}
\item $\PropReg{f,\F}$: for all $t\in \R\setminus
  f\Par{\critfV\cup\Sing{V}}$, the ideal $\IdealAngle{\F, f-t}$ is
  radical, equidimensional and $\Variety{\F,f-t}$ is either smooth of
  dimension $d-1$ or is empty.
\item $\PropPJSC{f,\F}$: there exists a non-empty
  Zariski open set $\NEZOSC\subset\C$ such that for all $t\in \R\cap
  \NEZOSC$, the restriction of $\ProjLeft{i-1}$ to $V\cap\Variety{f -
    t}\cap\VVA{f}{\F}{i}$ is proper for $1\leq i \leq d$.
\item $\PropPISSAC{f,\F}$: for any critical value $c$ of
  $\fV$ that is not isolated in $f\left(V\cap\R^n\right)$, there
  exists $x_c\in\VPCtot{f}{\F}$ such that $f\Par{x_c}=c$.
\end{itemize}

Assume that $\F$ satisfies assumptions $\assumptR$. We will prove that
up to a generic change of coordinates, the above properties are
satisfied. Properties $\PropReg{f,\F}$ and $\PropPJSC{f,\F}$ will be
used to prove that $\VVAtot{f}{\F}\setminus\Crit{f,V}$ has dimension
at most $1$ and $\VPCtot{f}{\F}$ is finite. This implies that the
first assertion in Definition \ref{def:propopt} holds, for the set $W(f,\F)$
defined in Section \ref{sec:propertiesoptimization}. They will also be
used to prove that the third assertion of Definition \ref{def:propopt}
is satisfied.

Finally, the second assertion of Definition \ref{def:propopt} will be
proved to be satisfied using Properties $\PropPJSC{f,\F}$ and
$\PropPISSAC{f,\F}$.

Theorem \ref{thm:generic} establishes that up to generic change of
coordinates, properties $\PropReg{f,\F}$, $\PropPJSC{f,\F}$ and
$\PropPISSAC{f,\F}$ hold.

\section{Algorithm}\label{sec:algo}

This section is structured as follows. After an outline of the
algorithm, we explain its specification. Next, we describe the
subroutines it uses and this section ends with a formal description of
the algorithm.

\subsection{Outline}\label{sec:algooutline}
There are basically two main steps in our algorithm. After a generic
change of coordinates, the first one is the computation of finite sets
from which a set containing all the local extrema of $\fV$ can be
recovered by deciding the emptiness of real algebraic sets. Reusing
the notations in the previous section, these sets are the following:
\begin{itemize}
\item a set $\VSP{\F}$ of sample points of $\Variety{\F}\cap\R^n$,
\item the set $\VPCtot{f}{\F}\cap\R^n$,
\item and the set of non-properness of the restriction to
  $\Variety{f-T}\cap\left(\VVAtot{f}{\F}\times\C\right)$ of the
  projection $\pi_T$.
\end{itemize}

Let $W=W(f, \F)$ be defined as the union of the above set of
non-properness, $f\left(\VSP{\F}\right)$ and
$f\left(\VPCtot{f}{\F}\cap\R^n\right)$. We prove in Section
\ref{sec:correctness} that Property $\PropOpt{W, V}$ holds. In
particular, this means that $W$ contains all the local extrema of
$\fV$ and is finite.

The second main step of the algorithm is then to detect the global
infimum among the values in $W$. By definition, $f^\star$ is the
smallest value $c$ in $V\cap\R^n$ such that
\begin{enumerate}
\item if $t<c$ then $t\not\in f\Par{V\cap\R^n}$ and
\item for all $t\geq c$, $\left[c,t\right]$ meets $f\Par{V\cap\R^n}$.
\end{enumerate}

Let $a_1<\dots<a_k$ be the values in $W$, let $a_0 = -\infty$ and
$a_{k+1}=+\infty$.

If $f^\star=-\infty$, then for any value $t\in ]-\infty, a_1[$,
$f^{-1}(t)\cap V\cap\R^n$ is not empty. This can be decided using any
algorithm for deciding the emptiness of real algebraic sets.

Now, let $0\leq i\leq k+1$ and assume that the infimum $f^\star$ is not
$a_0,\dots,a_{i-1}$. Since $W$ contains all the local extrema,
$f^\star\geq a_i$. If $a_i$ lies in $f\left(\VSP{\F}\right)$ or in
$f\left(\VPCtot{f}{\F}\cap\R^n\right)$ then this is a value attained
by $f$. In this case, $f^\star$ is necessarily $a_i$.

Else, if $a_i$ is an asymptotic value then there are values attained
by $f$ in every neighborhood of $a_i$. Since $f^\star\geq a_i$, the
third assertion of Property $\PropOpt{W, V}$ implies that for almost all
$t\in \left]a_{i}, a_{i+1}\right[$, $t$ is a value attained by $f$. In
particular, if $a_i$ is an asymptotic value then for a random number
$t\in \left]a_{i}, a_{i+1}\right[$, the variety
$\Variety{f-t}\cap\Variety{\F}\cap\R^n$ is non-empty. Thus, if
$\Variety{f-t}\cap\Variety{\F}\cap\R^n$ is not empty for a random
value $t$ in $\left]a_{i}, a_{i+1}\right[$ then $a_i$ is an asymptotic
value, that is necessarily $f^\star$. Else, $a_i$ is not relevant for
the optimization problem under consideration and we have $f^\star\geq
a_{i+1}$.

\subsection{Specifications}\label{section:spec}
In the descriptions of the algorithms, 
algebraic sets are represented with polynomial families that define
them and ideals are represented by a finite list of generators (e.g. a
Gröbner basis).

Let $Z\subset\R^n$ be a finite real algebraic set defined by
polynomials in $\Q\left[\bfX\right]$. It can be represented by a
rational parametrization, that is a sequence of polynomials
$q,q_0,q_1,\dots,q_n \in \Q\left[U\right]$ such that for all
$x=\Par{x_1,\dots,x_n}\in Z$, there exists $u\in\R$ such that
\[\left\{\begin{array}{ccc}
    q(u) & = & 0 \\
    x_1 & = & q_1(u)/q_0(u) \\
    & \vdots & \\
    x_n & = & q_n(u)/q_0(u)
  \end{array}\right.\]
with the convention that $q=1$ when $Z=\emptyset$. 
Moreover, a single point in $Z$ can be represented using isolating
intervals. Note that such a representation can be computed from a
Gr\"obner basis \cite{fglm, fm11, FGHR, RouillierRUR} and algorithms computing such a
representation are implemented in computer algebra systems.

Also, a real algebraic number $\alpha$ is represented by a
univariate polynomial $P$ and an isolating interval $I$.

\subsection{Subroutines}
In this paragraph we describe the main subroutines \SetContaining{}
and \FindInfimum{} that will are used in the main algorithm. They
correspond to the two main steps sketched in Section
\ref{sec:algooutline}.

We start with some standard subroutines on which these both
subroutines are based. Given a univariate polynomial $P$, we denote by
$\Roots{P}$ the set of its real roots.

\subsubsection*{\SamplePoints{}}
Given $\F\subset\Q\left[\bfX\right]$ satisfying assumptions $\bfR$,
\SamplePoints{} returns a list of equations $\listS\subset\PolRing$
such that $\Variety{\listS}$ contains at least one point in each
connected component of $\Variety{\F}\cap\R^n$. We refer to \cite{SaSc}
and references therein for an efficient algorithm performing this
task.

\subsubsection*{\SetOfNonProperness{}}
The routine \SetOfNonProperness{} takes as input $f\in\Q[\bfX]$ and
$\G\subset\Q[\bfX]$ such that $\dim\Variety{\G}\leq 1$. It returns a
univariate polynomial in $T$ whose set of roots contains the set of
non-properness of the restriction of $\Proj{T}$ to
$\Variety{f-T}\cap\left(\Variety{\G}\times\C\right)$.  Such an
algorithm is given in \cite{LaRo07, SaSc04}.

\subsubsection*{\RealRootIsolation}
Given $P\in\Q\left[T\right]$ whose set of real roots is
$a_1<\dots<a_k$, this routine returns a sorted list of $k$ pairwise
disjoint intervals with rational endpoints $\left[q_i,q_{i+1}\right]$
such that $a_i\in\left[q_i,q_{i+1}\right]$ (since the intervals are
disjoint, the list is sorted for the natural order : $[a,b] < [c,d]$
if and only if $b < c$). We refer to \cite{BaPoRo06, RouZim} for an
algorithm with this specification.

\subsubsection*{\IsEmpty}
Given $\G\subset\Q\left[\bfX\right]$ satisfying assumptions $\bfR$,
this routine returns either \textsf{true} if $\Variety{\G}\cap\R^n$ is
empty of \textsf{false} if it is nonempty. This is a weakened variant
of \SamplePoints{}.

\subsubsection*{\SetContaining{}}
It takes as input $f\in\PolRing$ and $\F\subset\PolRing$ satisfying
assumptions $\assumptR$, $\PropPJSC{f,\F}$, $\PropPISSAC{f,\F}$ and
$\PropReg{f,\F}$. We denote by $V$ the algebraic set defined by $\F$.

It returns a list
$\listS\subset\Q\left[\bfX\right]$, a list
$\listP\subset\Q\left[\bfX\right]$ and a polynomial
$\polNP\in\Q\left[T\right]$ such that the property
\[\PropOpt{f\Par{\Variety{\listS}}\cup
  f\Par{\Variety{\listP}}\cup\Roots{\polNP}, V}\] holds.
The list of polynomials $\listS$ and $\listP$ represent respectively
at least one point in each connected component of $V\cap \R^n$ and the
set of critical points of the restriction of the map $x\to f(x)$ to
$V$ and some fibers.
 

\vspace{1em}
\hrule
\vspace{0.5em}
\noindent$\SetContaining{\Par{f, \F}}$
\begin{itemize}[leftmargin=1.5em]
\item $\listS\leftarrow\SamplePoints\Par{\F}$;
\item $\polNP\leftarrow 1$;
\item for $1\leq i\leq d$ do
  \begin{itemize}
  \item $\listC[i]\leftarrow$ a list of equations defining
    $\ZariskiClosure{\VVA{f}{\F}{i}\setminus\crit{f,\Variety{\F}}}$;
  \item $\polNP\leftarrow$ the univariate polynomial
    $\polNP\times\SetOfNonProperness\Par{f, \listC[i]}$;
  \item $\listP[i]\leftarrow$ a list of equations defining
    $\VPC{f}{\F}{i}$.
  \end{itemize}
  \item return $\Par{\listS,\listP,\polNP}$;
\end{itemize}
\vspace{0.5em}
\hrule
\vspace{1em}

Its correctness is stated in Proposition
\ref{prop:SetContainingcorrect}. Its proof relies on intermediate
results given in Section \ref{sec:SetContainingcorrect}.

\subsubsection*{\FindInfimum{}}
The routine \FindInfimum{} takes as input:
\begin{itemize}
\item $f\in\Q[\bfX]$,
\item $\F\subset\Q[\bfX]$ satisfying assumptions $\assumptR$ and
  $\PropReg{f,\F}$; we let $V\subset \C^n$ be the algebraic set it
  defines,
\item $\listS\subset\Q\left[\bfX\right]$,
  $\listP\subset\Q\left[\bfX\right]$ and $\polNP\in\Q\left[T\right]$
  such that $\PropOpt{f\Par{\Variety{\listS}}\!\cup\!
  f\Par{\Variety{\listP}}\cup\Roots{\polNP}, V}$ holds.
\end{itemize}

It returns
\begin{itemize}
\item $+\infty$ if $\Variety{\F}\cap\R^n$ is empty;
\item $-\infty$ if $f$ is not bounded below on $\Variety{\F}\cap\R^n$;
\item if $f^\star$ is finite and not reached:
  $\polNP\in\Q\left[T\right]$ and an interval $I$ isolating $f^\star$;
\item if $f^\star$ is reached, a rational parametrization encoding
  $x^\star$ and $f^\star = f\Par{x^\star}$.
\end{itemize}


\vspace{1em}
\hrule
\vspace{0.5em}
\noindent$\FindInfimum{\Par{f,\F,\listS,\listP,\polNP}}$
\begin{itemize}[leftmargin=1.5em]
  \item $a_1<\!\cdots\!<a_k\leftarrow f\Par{\Variety{\listS}}\cup f\Par{\Variety{\listP}}\cup\Roots{\polNP}$;
  \item $a_{k+1}= +\infty$;
  \item $q_0\leftarrow$ a random rational $<a_1$;
  \item if \IsEmpty$\left(\left\{ f-q_0,\F\right\} \right)$=false then
    \begin{itemize}[leftmargin=1em]
    \item return $-\infty$; 
    \end{itemize}
  \item $i\leftarrow 1$;
  \item while $i\leq k$ do
  \begin{itemize}[leftmargin=1em]
  \item if $a_i\in f\Par{\Variety{\listS}}\cup f\Par{\Variety{\listP}}$ then
    \begin{itemize}
    \item RP $\leftarrow$ a rational parametrization encoding a
      minimizer $x^\star$ and $f\left(x^\star\right)=a_i$;
    \item return RP
    \end{itemize}
    else
    \begin{itemize}
    \item $q_i\leftarrow$ a random rational in $\left]a_i,a_{i+1}\right[$;
    \item if \IsEmpty$\left(\left\{ f-q_i,\F\right\} \right)$=false then
      \begin{itemize}
      \item return $\left(\polNP,\left] q_{i-1}, q_i  \right[\right)$
      \end{itemize}
      else
      \begin{itemize}
      \item $i\leftarrow i+1$
      \end{itemize}
    \end{itemize}
  \end{itemize}
  \item return $a_{k+1}$
\end{itemize}
\vspace{0.5em}
\hrule
\vspace{1em}

Its correctness is stated by Proposition \ref{prop:FindInfimumcorrect}
whose proof is in Section \ref{prop:FindInfimumcorrect}.

By assumption on the inputs, $\Variety{\listS}\cup
\Variety{\listP}$ is finite. As explained in Section
\ref{section:spec}, a single point $x$ that lies in this variety can
be represented by a rational parametrization $q,q_0,q_1,\dots,q_n$ and
an interval isolating the corresponding root of $q$. From this
parametrization and the isolating interval, an interval isolating
$f\Par{x}$ can be computed. Likewise, the roots of $\polNP$ are
represented by isolating intervals. These intervals can be computed
such that they do not intersect. Hence, they can be sorted so that the
$i$-th interval corresponds to $a_i$.

Then, testing whether $a_i\in f\Par{\Variety{\listS}}\cup
f\Par{\Variety{\listP}}$ is done by testing whether the interval
corresponding with $a_i$ comes from the parametri\-zation of
$\Variety{\listS}\cup \Variety{\listP}$. If so, the parametrization
and the isolating interval of $q$ corresponding with $a_i$ are an
encoding for $a_i$ and a point $x_{a_i}$ such that
$f\Par{x_{a_i}}=a_i$.

\subsection{Main Algorithm}


The main routine \Optimize{} takes as input $f \in \Q[\bfX]$ and
$\F\subset\Q[\bfX]$ satisfying assumptions $\assumptR$. It returns
\begin{itemize}
\item $+\infty$ if $\Variety{\F}\cap\R^n$ is empty;
\item $-\infty$ if $f$ is not bounded below on $\Variety{\F}\cap\R^n$;
\item if $f^\star$ is finite and not reached:
  $\polNP\in\Q\left[T\right]$ and an interval $I$ isolating $f^\star$;
\item if $f^\star$ is reached, a rational parametrization encoding
  $x^\star$ and $f^\star = f\Par{x^\star}$.
\end{itemize}

\vspace{1em}
\hrule
\vspace{0.5em}
\noindent
$\Optimize{\Par{f,\F}}$.
\begin{itemize}[leftmargin=1.5em]
\item $\bfA\leftarrow$ a random matrix in $\GLnQ$;
\item $\Par{\listS, \listP,\polNP}\leftarrow
  \SetContaining\Par{f^\bfA,\F^\bfA}$;
\item \textsf{Infimum} $\leftarrow$ \FindInfimum$\Par{f^\bfA,\F^\bfA,
    \listS, \listP, \polNP}$;
\item return \textsf{Infimum}.
\end{itemize}
\vspace{0.5em}
\hrule

  

  
    

\section{Proof of correctness}\label{sec:correctness}
We first assume the following theorem, it is proved in Section
\ref{sec:generic}. It states that the properties $\PropReg{f, \F}$,
$\PropPJSC{f, \F}$ and $\PropPISSAC{f, \F}$ defined in Section
\ref{sec:genericnotations} are satisfied up to a generic change of
coordinates.  

\begin{theorem}\label{thm:generic}
  Let $f\in\PolRing$ and $\F\subset\PolRing$ satisfying assumptions
  $\assumptR$. There exists a non-empty Zariski open set
  $\NEZOS\subset\GLnC$ such that for all $\bfA\in\GLnQ\cap\NEZOS$, the
  properties $\PropReg{f^\bfA,\F^\bfA}$, $\PropPJSC{f^\bfA,\F^\bfA}$
  and $\PropPISSAC{f^\bfA,\F^\bfA}$ hold.
\end{theorem}

Let $\NEZOS\subset\GLnC$ be the Zariski open set given in Theorem
\ref{thm:generic}. We prove in the sequel that if the random matrix
chosen in $\Optimize$ lies in $\NEZOS$ then $\Optimize$ is correct.

The correctness of \Optimize{} is a consequence of the correctness of
the subroutines \SetContaining{} and \FindInfimum{}. The correctness
of \SetContaining{} is given in Section \ref{sec:SetContainingcorrect}
below. The proof of correctness of \FindInfimum{} is given in Section
\ref{sec:FindInfimumcorrect} page \pageref{sec:FindInfimumcorrect}.

\subsection{Correctness of \SetContaining{}}\label{sec:SetContainingcorrect}
We first state the correctness of $\SetContaining\Par{f^\bfA,\F^\bfA}$.
\begin{proposition}\label{prop:SetContainingcorrect}
  Let $f\in\PolRing$ and
  $\F=\left\{f_1,\dots,f_\nV\right\}\subset\PolRing$ satisfying
  assumptions $\assumptR$. Let $\NEZOS\subset\GLnC$ be the
  Zariski open set defined in Theorem \ref{thm:generic}. Then for all
  $\bfA\in\GLnQ\cap\NEZOS$, $\SetContaining\Par{f^\bfA,\F^\bfA}$ returns a list
$\listS\subset\Q\left[\bfX\right]$, a list
$\listP\subset\Q\left[\bfX\right]$ and a polynomial
$\polNP\in\Q\left[T\right]$ such that the property
\[\PropOpt{f^\bfA\Par{\Variety{\listS}}\cup
  f^\bfA\Par{\Variety{\listP}}\cup\Roots{\polNP}, V^\bfA}\] holds.
\end{proposition}

Given $\bfA\in\GLnQ$, let $W^\bfA$ be the set of values
\[W^\bfA\!=\! f^\bfA\left(\VSP{\F^\bfA}\right)\cup
f^\bfA\left(\VPCtot{f^\bfA}{\F^\bfA}\right)\cup\NP{\Proj{T}}{\ZariskiClosure{\VVAtot{f^\bfA}{\F^\bfA}\setminus\crit{f^\bfA,V^\bfA}}}\subset\C\]
with $\VSP{\F^\bfA}=V\Par{\listS}$ and
$\VPCtot{f^\bfA}{\F^\bfA}=V\Par{\listP}$ and
$\Roots{\polNP}=\NP{\Proj{T}}{\ZariskiClosure{\VVAtot{f^\bfA}{\F^\bfA}\setminus\crit{f^\bfA,V^\bfA}}}$.

We start by establishing that $\PropOpt{W^\bfA, V^\bfA}$ holds: we
prove below that $W^\bfA$ contains all the local extrema (Proposition
\ref{prop:localextr}), that it is finite (Proposition
\ref{prop:finite}) and the last assertion in Definition
\ref{def:propopt} (Proposition \ref{prop:topo}). Finally, we conclude
with the proof of Proposition \ref{prop:SetContainingcorrect} page
\pageref{proof:SetContainingcorrect}.

Since $V^\bfA$ is an algebraic variety, the image
$f^\bfA\Par{V^\bfA\cap\R^n}$ is a semi-algebraic subset of
$\R$. Hence, it is a finite union of real disjoint intervals. They are
either of the form $\left[b_i,b_{i+1}\right]$,
$\left[b_i,b_{i+1}\right[$, $\left]b_i,b_{i+1}\right]$ or
$\left\{b_i\right\}$, for some $b_0\in\R\cup\left\{-\infty\right\}$,
$b_1,\dots, b_r\in\R$ and $b_{r+1}\in \R\cup\{+\infty\}$. Then the
local extrema of $\fVbfA$ are the $b_i$'s. If $b_i$ is an endpoint
included in the interval, then it is reached, meaning that it is
either a minimum or a maximum. If the interval is a single point then
$b_i$ is isolated in $f^\bfA\Par{V^\bfA\cap\R^n}$. Else, it is not
isolated. If $b_i$ is an endpoint that is not included in the
interval, then $b_i\not\in f^\bfA\Par{V^\bfA\cap\R^n}$ is an extremum
that is not reached.  Remark that our goal is to find $b_0$, that is
equal to $f^\star$.

\begin{proposition}\label{prop:localextr}
  Let $\NEZOS\subset\GLnC$ be the Zariski open set defined in Theorem
  \ref{thm:generic}. For all $\bfA\in\GLnQ\cap\NEZOS$, and any local
  extremum $\locextr\in\R$ of $\fVbfA$, the following holds.
  \begin{enumerate}
  \item If $\locextr$ is a value that is isolated in
    $f^\bfA\Par{V^\bfA\cap\R^n}$ then $\locextr\in
    f^\bfA\left(\VSP{\F^\bfA}\right)$;
  \item if $\locextr$ is a value that is not isolated in
    $f^\bfA\Par{V^\bfA\cap\R^n}$ such that there exists $x_\locextr\in
    V^\bfA\cap\R^n$ with $f^\bfA\Par{x_\locextr} = \locextr$ then
    $\locextr\in f^\bfA\left(\VPCtot{f^\bfA}{\F^\bfA}\right)$;
  \item if $\locextr\not\in f^\bfA\Par{V^\bfA\cap\R^n}$ then
    $\locextr\in\NP{\Proj{T}}{\ZariskiClosure{\VVAtot{f^\bfA}{\F^\bfA}\setminus\crit{f^\bfA,V^\bfA}}}$.
  \end{enumerate}
  As a consequence, every local extremum of
  $\fVbfA$ is contained in $W^\bfA$.
\end{proposition}
\begin{proof}
  Let $\locextr\in\R$ be a local extremum.
  \paragraph{Case 1}
  Since $\locextr$ is isolated, there exists $x_\locextr\in
  V^\bfA\cap\R^n$ such that
  $f^\bfA\left(x_\locextr\right)=\locextr$. Let $C^\bfA$ be the
  connected component of $V^\bfA\cap\R^n$ containing $x_\locextr$. We
  prove that $f^\bfA$ is constant on $C^\bfA$. Let $x'\in C^\bfA$ and
  assume that $f^\bfA\left(x'\right)\neq \locextr$. Since $\locextr$
  is isolated, there exists a neighborhood $\mathcal{B}$ of $\locextr$
  such that $f^\bfA\left(C^\bfA\right)$ is the union of
  $\left\{\locextr\right\}$ and some set $S$ that contains
  $f^\bfA\left(x'\right)$ but  does not meet $\mathcal{B}$. In
  particular, $f^\bfA\left(C^\bfA\right)$ is not connected. This is a
  contradiction since $f^\bfA$ is continuous and $C^\bfA$ connected.

  The set $\VSP{\F^\bfA}$ is a set containing at least one point in
  each connected component of $V^\bfA\cap\R^n$. In particular it
  contains a point $y$ in the connected component $C^\bfA$ of
  $x_\locextr$. Since the restriction of $f^\bfA$ to $C^\bfA$ is
  constant, $f^\bfA\left(y\right)=\locextr$, so that $\locextr\in
  f^\bfA\left(\VSP{\F^\bfA}\right)$.

  \paragraph{Case 2}
  Since $\bfA\in\GLnQ\cap\NEZOS$, property
  $\PropPISSAC{f^\bfA,\F^\bfA}$ holds (Theorem \ref{thm:generic}).
  This means that there exists $x_\locextr\in\VPCtot{f^\bfA}{\F^\bfA}$
  such that $f^\bfA\Par{x_\locextr}=\locextr$, that is $\locextr\in
  f^\bfA\left(\VPCtot{f^\bfA}{\F^\bfA}\right)$.

  \paragraph{Case 3}
  If $\locextr\not\in f^\bfA\Par{V^\bfA\cap\R^n}$, by definition, as a
  local extremum, there exists a closed neighborhood $\ClosedN$ of
  $\locextr$ such that we can construct a sequence
  $(\suitex{k})_{k\in\N}\subset
  \left(f^\bfA\right)^{-1}\left(\ClosedN\right)\cap V^\bfA\cap\R^n$
  such that $f^\bfA\left(\suitex{k}\right)\rightarrow \locextr$. We
  first prove that we can not extract a converging subsequence from
  $(\suitex{k})$. Indeed, assume that there exists a converging
  subsequence $(x'^{(k)})$ and denote by $x$ its limit. Since
  $V^\bfA\cap\R^n$ and
  $\left(f^\bfA\right)^{-1}\left(\ClosedN\right)\cap\R^n$ are closed
  sets for the euclidean topology, $x$ lies in
  $\left(f^\bfA\right)^{-1}\left(\ClosedN\right)\cap V^\bfA\cap\R^n$.

  As a subsequence of $f^\bfA\left(\suitex{k}\right)$, the sequence
  $f^\bfA\left(x'^{(k)}\right)$ tends to $\locextr$. Moreover, by
  continuity of $f^\bfA$, $f^\bfA\left(x'^{(k)}\right)$ tends to
  $f^\bfA\left(x\right)$. This implies that $f^\bfA(x)=\locextr$, that
  is $\locextr$ is attained, which is a contradiction. Since this is
  true for any converging subsequence $(x'^{(k)})$ of $(\suitex{k})$,
  this implies that $(\suitex{k})$ can not be bounded. Finally, this
  proves that $\|(\suitex{k})\|$ tends to $\infty$.

  Let $\eps>0$. There exists $k_0\in\N$ such that for all $k\geq k_0$,
  $f^\bfA\left(\suitex{k}\right)\in\left[\locextr-\eps,\locextr+\eps\right]$. By
  construction of $\suitex{k}$,
  $\left(f^\bfA\right)^{-1}\left(f^\bfA\left(\suitex{k}\right)\right)\cap
  V^\bfA\cap\R^n\neq\emptyset$.

  By assumption, we have $\bfA\in\NEZOS$; then Theorem
  \ref{thm:generic} implies that $\PropReg{\F^\bfA}$ and
  $\PropPJSC{f^\bfA,\F^\bfA}$ hold. Thus \cite[Proposition
  1.3]{GrGuSaZh} ensures that for all $t\in \R\cap \NEZOSC^\bfA$,
  $V^\bfA\cap \Variety{f^\bfA-t}\cap\R^n$ is empty if and only if
  $\VVAtot{f^\bfA}{\F^\bfA}\cap\Variety{f^\bfA-t}\cap\R^n$ is empty.
  Since $\R\setminus\NEZOSC^\bfA$ is finite, one can assume without
  loss of generality that for all $k$,
  $f^\bfA\left(x^{(k)}\right)\in\NEZOSC^\bfA$.
  
  Then
  $\Par{f^\bfA}^{-1}\Par{f^\bfA\Par{\suitex{k}}}\cap\VVAtot{f^\bfA}{\F^\bfA}\cap\R^n\neq\emptyset$. Picking
  up a point $\widetilde{x}^{(k)}$ in this last set, for each $k\geq
  k_0$, leads to the construction of a sequence of points
  $\left(\widetilde{x}^{(k)}\right)$ in
  $\VVAtot{f^\bfA}{\F^\bfA}\cap\R^n$ such that
  $f\left(\widetilde{x}^{(k)}\right)$ tends to $\locextr$. Since
  $\VVAtot{f^\bfA}{\F^\bfA}\subset V^\bfA$ and $\locextr$ is not
  reached, this sequence is unbounded. Since
  $f^\bfA\left(\crit{f^\bfA,V^\bfA}\right)$ is finite by Sard's
  theorem, one can assume without loss of generality that for all $k$,
  $\widetilde{x}^{(k)}\not\in \crit{f^\bfA,V^\bfA}$.

  Then considering the sequence
  $\Par{\widetilde{x}^{(k)},t=f^\bfA\Par{\widetilde{x}^{(k)}}}$ proves
  that $\Proj{T}$ restricted to
  $\Variety{f^\bfA-T}\cap\left(\ZariskiClosure{\VVAtot{f^\bfA}{\F^\bfA}\setminus\crit{f^\bfA,\F^\bfA}}\times\C\right)$
  is not proper at $\locextr$; in other words
  $\locextr\in\NP{\Proj{T}}{\ZariskiClosure{\VVAtot{f^\bfA}{\F^\bfA}\setminus\crit{f^\bfA,\F^\bfA}}}$.
\end{proof}

\begin{proposition}\label{prop:finite}
  For all $\bfA\in\GLnQ\cap\NEZOS$, the set $W^\bfA$ is finite.
\end{proposition}
\begin{proof}
  Recall that, by definition, $W^\bfA = f^\bfA\Par{\VSP{\F^\bfA}}\cup
  f^\bfA\Par{\VPCtot{f^\bfA}{\F^\bfA}}\cup\NP{\Proj{T}}{\VVAtot{f^\bfA}{\F^\bfA}}$. We
  prove below the following assertions.
  \begin{enumerate}
  \item $\VSP{\F^\bfA}$ is finite,
  \item For $1\leq i\leq d$, $\VPC{f^\bfA}{\F^\bfA}{i}$ is finite and
  \item $\NP{\Proj{T}}{\ZariskiClosure{\VVAtot{f^\bfA}{\F^\bfA}\setminus\crit{f^\bfA,\F^\bfA}}}$ is finite.
  \end{enumerate}
  
  \paragraph{Assertion 1}
  The first assertion is true for all $\bfA$, since by assumption,
  $\VSP{\F^\bfA}$ is a finite set.
  
  \paragraph{Assertion 2}
  
  Let $1\leq i\leq d$. Recall that
  \[\VPC{f^\bfA}{\F^\bfA}{i}=\ZariskiClosure{\VVA{f^\bfA}{\F^\bfA}{i}\setminus\critfVbfA}\cap\critfVbfA.\]
  We first prove that
  $\ZariskiClosure{\VVA{f^\bfA}{\F^\bfA}{i}\setminus\critfVbfA}$ has
  dimension $1$. Next, it will be easy to deduce that its intersection
  with $\critfVbfA$ has dimension at most $0$.

  By assumption, we have $\bfA\in\NEZOS$. Thus, Theorem
  \ref{thm:generic} implies that $\PropReg{\F^\bfA}$ and
  $\PropPJSC{f^\bfA,\F^\bfA}$ hold. Thus \cite[Proposition
  1.3]{GrGuSaZh} ensures that for all $t\in\NEZOSC^\bfA$, the
  algebraic set $\Variety{f^\bfA-t}\cap\VVA{f^\bfA}{\F^\bfA}{i}$ has
  dimension at most zero.

  Now, let $Z^\bfA$ be an irreducible component of
  $\ZariskiClosure{\VVA{f^\bfA}{\F^\bfA}{i}\setminus\critfVbfA}$. In
  particular, $Z^\bfA$ is an irreducible component of
  $\VVA{f^\bfA}{\F^\bfA}{i}$ that is not contained in
  $\critfVbfA$. Consider the restriction $\fZbfA\colon Z^\bfA
  \longrightarrow\C$. Its image has a Zariski closure of dimension $0$
  or $1$.

  Assume first that $f^\bfA\left(Z^\bfA\right)$ is
  $0$-dimensional. Then as a continuous function, $\fZbfA$ is locally
  constant. This implies that $Z^\bfA$ is contained in the critical
  locus of $\fVCbfA$. In particular, this means that
  $Z^\bfA\subset\critfVbfA$, which is a contradiction.

  Then all irreducible components $Z^\bfA$ are such that
  $\ZariskiClosure{f^\bfA\Par{Z^\bfA}}$ has dimension $1$.  From the
  Theorem on the dimension of fibers (\cite[Theorem 7, Chapter 1,
  pp. 76]{Shafarevich77}), there exists a Zariski open set
  $U\subset\C$ such that for all $y\in U$, $\dim
  \left(f^\bfA\right)^{-1}=\dim Z^\bfA -1$. In particular if $t$ lies
  in the non-empty Zariski open set $U\cap\NEZOSC^\bfA$, the following
  holds
  \[0\geq \dim \left(f^\bfA\right)^{-1}=\dim Z^\bfA -1.\]

  Then every irreducible component $Z^\bfA$ of
  $\ZariskiClosure{\VVA{f^\bfA}{\F^\bfA}{i}\setminus\critfVbfA}$ has
  dimension $\leq 1$, so that
  $\dim\ZariskiClosure{\VVA{f^\bfA}{\F^\bfA}{i}\setminus\critfVbfA}\leq
  1$.

  Now, let $Z^\bfA_1\cup\dots\cup Z^\bfA_\alpha\cup\dots\cup
  Z^\bfA_\beta$ be the decomposition of $\VVA{f^\bfA}{\F^\bfA}{i}$ as
  a union of irreducible components. Up to reordering, assume that
  \begin{itemize}
  \item for $1\leq i\leq \alpha$, $Z^\bfA_i\not\subset\critfVbfA$,
  \item for $\alpha+1\leq j\leq \beta$, $Z^\bfA_i\subset\critfVbfA$.
  \end{itemize}
  Then the decomposition of
  $\ZariskiClosure{\VVA{f^\bfA}{\F^\bfA}{i}\setminus\critfVbfA}$ as a
  union of irreducible components is $Z^\bfA_1\cup\dots\cup
  Z^\bfA_\alpha$.

  Let $1\leq i\leq \alpha$ and consider $Z^\bfA_i\cap\critfVbfA$. If
  it is non-empty, since $Z^\bfA_i\not\subset\critfVbfA$,
  \cite[Corollary 3.2 p. 131]{kunz} implies that $Z^\bfA_i\cap\critfVbfA$
  has dimension less than or equal to $\dim Z^\bfA_i -1\leq
  1-1=0$. Finally, this proves that
  $\ZariskiClosure{\VVA{f^\bfA}{\F^\bfA}{i}\setminus\critfVbfA}\cap\critfVbfA$
  has dimension $\leq 0$.

  \paragraph{Assertion 3}
  Since $\bfA\in\NEZOS$,  Theorem \ref{thm:generic} implies that 
  $\PropReg{\F^\bfA}$ and $\PropPJSC{f^\bfA,\F^\bfA}$ hold. We have
  proved above, that for $1\leq i\leq d$,
  $\ZariskiClosure{\VVA{f^\bfA}{\F^\bfA}{i}\setminus\critfVbfA}$ has
  dimension at most $1$. 

  Then by \cite[Theorem 3.8]{jelonek}, each set
  $\NP{\Proj{T}}{\ZariskiClosure{\VVA{f^\bfA}{\F^\bfA}{i}\setminus\crit{f^\bfA,\F^\bfA}}}$
  has dimension at most $0$, thus
  $\NP{\Proj{T}}{\ZariskiClosure{\VVAtot{f^\bfA}{\F^\bfA}\setminus\crit{f^\bfA,\F^\bfA}}}$
  is finite.

  To conclude, we prove that the set of values $t\in \C$ such that
  there exists a sequence $(\suitex{k})_{k\in
    \N}\subset\ZariskiClosure{\VVAtot{f^\bfA}{\F^\bfA}\setminus\crit{f^\bfA,\F^\bfA}}$
  satisfying $\displaystyle \lim_{k\rightarrow
    +\infty}||\suitex{k}||=+\infty$ and
  $\displaystyle\lim_{k\rightarrow+\infty} f^\bfA(\suitex{k})=t$ is
  exactly the set 
  $\NP{\Proj{T}}{\ZariskiClosure{\VVAtot{f^\bfA}{\F^\bfA}\setminus\crit{f^\bfA,\F^\bfA}}}$.

  Let $t_0\in\C$ and
  $\Par{\suitex{k}}=\left(x_1^{(k)},\dots,x_n^{(k)}\right)$ be a
  sequence of points in the set
  $\ZariskiClosure{\VVAtot{f^\bfA}{\F^\bfA}\setminus\crit{f^\bfA,\F^\bfA}}$
  such that $\displaystyle \lim_{k\rightarrow
    +\infty}\left\|\suitex{k}\right\|=+\infty$ and
  $\displaystyle\lim_{k\rightarrow+\infty}
  f^\bfA\Par{\suitex{k}}=t_0$.

  Let $\eps>0$. There exists $N\in\N$ such that for all $k\geq N$,
  $\left|f^\bfA\left(\suitex{k}\right)-t_0\right|\leq \eps$. In
  particular, for all $k\geq N$,
  $\left(f^\bfA\right)\left(\suitex{k}\right)$ lies in the closed ball
  $\ClosedBall{t_0,\eps}$. This means that
  $\Proj{T}^{-1}\Par{\ClosedBall{t_0,\eps}}\cap\Variety{f^\bfA-T}\cap
  \left(\ZariskiClosure{\VVAtot{f^\bfA}{\F^\bfA}\setminus\crit{f^\bfA,\F^\bfA}}\times\C\right)$
  contains all the points
  \[\Par{x_1^{(k)},\dots,x_n^{(k)}, t= f^\bfA\Par{\suitex{k}}}\]
  for $k\geq N$. Since $\Par{\suitex{k}}$ is not bounded, we deduce
  that
  \[\Proj{T}^{-1}\Par{\ClosedBall{t_0,\eps}\cap\Variety{f^\bfA-T}\cap
    \left(\ZariskiClosure{\VVAtot{f^\bfA}{\F^\bfA}\setminus\crit{f^\bfA,\F^\bfA}}}\times\C\right)\]
  is not bounded.  This means that $t_0$ is a point where the
  projection $\Proj{T}$ restricted to $\Variety{f^\bfA-T}\cap
  \left(\ZariskiClosure{\VVAtot{f^\bfA}{\F^\bfA}\setminus\crit{f^\bfA,\F^\bfA}}\times\C\right)$
  is not proper.

  Conversely, if $t_0\in\C$ is such that for all $\eps>0$,
  \[\Proj{T}^{-1}\Par{\ClosedBall{t_0,\eps}\cap\Variety{f^\bfA-T}\cap
    \left(\ZariskiClosure{\VVAtot{f^\bfA}{\F^\bfA}\setminus\crit{f^\bfA,\F^\bfA}}}\times\C\right)\]
    is not bounded, we can construct by induction a sequence
    $\Par{\Par{\suitex{k},f^\bfA\Par{\suitex{k}}}}_{k\in\N}$ such
    that:
  \begin{itemize}
  \item for all $k\in\N$, the point
    $\Par{\suitex{k},f^\bfA\Par{\suitex{k}}}$ lies in
    \[\Proj{T}^{-1}\!\!\Par{\ClosedBall{t_0,\frac1k}\!\cap\!\Variety{f^\bfA-T}\cap
      \left(\ZariskiClosure{\VVAtot{f^\bfA}{\F^\bfA}\setminus\crit{f^\bfA,\F^\bfA}}}\times\C\right);\]
  \item for all $k\in\N$, $\left\|x_{k+1}\right\|>2\left\|\suitex{k}\right\|$.
  \end{itemize}
  In particular, $\Par{\suitex{k}}_{k\in
    \N}\subset\ZariskiClosure{\VVAtot{f^\bfA}{\F^\bfA}\setminus\crit{f^\bfA,\F^\bfA}}$,
  $\displaystyle \lim_{k\rightarrow
    +\infty}\left\|\suitex{k}\right\|=+\infty$ and
  $\displaystyle\lim_{k\rightarrow+\infty} f^\bfA(\suitex{k})=t_0$.
\end{proof}

\begin{proposition}\label{prop:topo}
  For all $\bfA\in\GLnQ\cap\NEZOS$, let
  $W^\bfA=\left\{a_1,\dots,a_k\right\}$, $a_0=-\infty$ and
  $a_{k+1}=+\infty$ with $a_i< a_{i+1}$ for $0\leq i \leq k$. There
  exists a non-empty Zariski open set $\NEZOSC^\bfA\subset\C$ such
  that for all $0\leq i \leq k$ and all couples $(t, t')$ in
  $\left]a_i,a_{i+1}\right[$
\[
f^{-1}\left( t \right)\cap V\cap\R^n=\emptyset\Longleftrightarrow f^{-1}\left( t' \right)\cap V\cap\R^n=\emptyset.
\]
\end{proposition}
\begin{proof}
  Our proof is by contradiction. Assume that there exists $i$ such
  that there exists $a\in\left]a_i,a_{i+1}\right[\cap\NEZOSC^\bfA$
  such that ${f^\bfA}^{-1}\Par{a}\cap V^\bfA\cap\R^n=\emptyset$ and
  $b\in\left]a_i,a_{i+1}\right[\cap\NEZOSC^\bfA$ such that
  ${f^\bfA}^{-1}\Par{b}\cap V^\bfA\cap\R^n\neq\emptyset$. Without
  loss of generality, we can assume that $a<b$ and
  \[b=\inf\left\{t\in\left]a_i,a_{i+1}\right[\cap\NEZOSC^\bfA \text{
      s.t.} {f^\bfA}^{-1}\Par{t}\cap
    V^\bfA\cap\R^n\neq\emptyset\right\}.\] Then $b$ is a local infimum
  of $\fVbfA$ that is neither $a_i$ nor $a_{i+1}$. However, according
  to Proposition \ref{prop:localextr}, $b$ lies in $W^\bfA$. Hence
  there exists $i$ such that $b=a_i$, which is a contradiction.
\end{proof}

We can now give the proof of Proposition \ref{prop:SetContainingcorrect} using the above propositions. 

\begin{proof}
  Let $\listS\subset\Q\left[\bfX\right]$,
  $\listP\subset\Q\left[\bfX\right]$ and $\polNP\in\Q\left[T\right]$
  be the output of $\SetContaining{\Par{f,\F}}$. Denote by $W$ the set
  \[f\Par{\Variety{\listS}}\cup
  f\Par{\Variety{\listP}}\cup\Roots{\polNP}.\] The routine
  \SetContaining{} is correct if property $\PropOpt{W, V}$ holds (see
  Definition \ref{def:propopt}). Then we check that
  \begin{enumerate}
  \item every local extremum of $\fV$ is contained in $W$,
  \item $W$ is finite,\label{proof:SetContainingcorrect}
  \item let $W=\left\{a_1,\dots,a_k\right\}$, $a_0=-\infty$ and
    $a_{k+1}=+\infty$ with $a_i < a_{i+1}$ for $0\leq i \leq k$. There
    exists a non-empty Zariski open set $\NEZOSC\subset\C$ such that
    for all $0\leq i \leq k$ and all couples $(t, t')$ in
    $\left]a_i,a_{i+1}\right[$
\[
f^{-1}\left( t \right)\cap V\cap\R^n=\emptyset\Longleftrightarrow f^{-1}\left( t' \right)\cap V\cap\R^n=\emptyset. 
\]
  \end{enumerate}
  Proposition \ref{prop:localextr} establishes the assertion
  $1$. Assertion $2$ is a restatement of Proposition
  \ref{prop:finite}. Assertion $3$ is established by Proposition
  \ref{prop:topo}.
\end{proof}

\subsection{Correctness of \FindInfimum{}}\label{sec:FindInfimumcorrect}
Finally, we prove the correctness of the routine $\FindInfimum$.
\begin{proposition}\label{prop:FindInfimumcorrect}
  Let $\bfA\in\GLnQ\cap\NEZOS$, $f\in\PolRing$, $\F \subset\PolRing$
  satisfying assumptions $\assumptR$, $\listS^\bfA\subset\PolRing$,
  $\listP^\bfA\subset\PolRing$ and $\polNP^\bfA\in\Q\left[T\right]$.
  Let $W^\bfA=\left\{a_1,\dots,a_k\right\}$, be the finite algebraic
  set \[f^\bfA\Par{\Variety{\listS^\bfA}}\cup
  f\Par{\Variety{\listP^\bfA}}\cup\Roots{\polNP^\bfA},\] and
  assume that $\PropOpt{W^\bfA, V^\bfA}$ is satisfied. Then let $a_0=-\infty$,
  $a_{k+1}=+\infty$ and let $\NEZOSC^\bfA\subset\C$ be the
  Zariski open set such that, for all $0\leq i \leq k$ all
    couples $(t, t')$ in $\left]a_i,a_{i+1}\right[$
\[
f^{-1}\left( t \right)\cap V\cap\R^n=\emptyset\Longleftrightarrow f^{-1}\left( t' \right)\cap V\cap\R^n=\emptyset. 
\]

  If the random rational numbers computed in \FindInfimum{} lie in
  $\NEZOSC^\bfA$ then \FindInfimum{} is correct.
\end{proposition}
\begin{proof}
  Since $\bfA\in\NEZOS$, Theorem \ref{thm:generic} implies that
  property $\PropReg{f^\bfA,\F^\bfA}$ is satisfied. Hence \IsEmpty{}
  is always called with a correct input.

  Assume first that $f^\star=-\infty$.  By assertion 3 of
  $\PropOpt{W^\bfA, V^\bfA}$ (see Definition \ref{def:propopt}), the
  fiber of $f^\bfA$ at a rational $q_0\in\NEZOSC^\bfA\cap \Q$ such that
  $q_0<a_1$ is not empty. Hence the first call to \IsEmpty{} returns
  false so that \FindInfimum{} returns $-\infty$. Now, remark that if
  \FindInfimum{} returns $-\infty$, then assertion 3 of
  $\PropOpt{W^\bfA, V^\bfA}$ implies that $f^\star=-\infty$.

  If $f^\star$ is finite, because the second assertion of
  $\PropOpt{W^\bfA, V^\bfA}$ holds, it is sufficient to identify the
  smallest local extremum of $\fVbfA$ in $W^\bfA$. To this end, we
  want to detect an eventual redundant value in $W^\bfA$. Such a
  redundant value, say $a_i$, is such that the interval
  $\left[a_i,a_{i+1}\right[$ does not contain any value reached by
  $f^\bfA$. In particular, it is a value that is not in
  $f^\bfA\Par{\Variety{\listS^\bfA}}\cup f\Par{\Variety{\listP^\bfA}}$
  and such that $\fVbfA$ does not reach any value in the interval
  $\left]a_i,a_{i+1}\right[$. Because of assertion 3 of
  $\PropOpt{W^\bfA, V^\bfA}$, testing this last point is equivalent to
  test the emptiness of the real fiber of $f^\bfA$ at some rational
  $q_i\in\NEZOSC^\bfA\cap\left]a_i,a_{i+1}\right[\cap \Q$.

  If $\Variety{\F}\cap\R^n$ is empty then so are the varieties
  $\Variety{\listS^\bfA}\cap\R^n$ and
  $\Variety{\listP^\bfA}\cap\R^n$. Since $\Variety{\F}\cap\R^n$ is
  empty, each call to the routine \IsEmpty{} in the loop returns
  false. Hence, the algorithm leaves the loop without returning any
  value, so that $a_{k+1}=+\infty$ is returned.

  Finally, this proves that the routine \FindInfimum{} is correct.
\end{proof}

\section{Proof of Theorem \ref{thm:generic}}\label{sec:generic}
This section is devoted to prove Theorem \ref{thm:generic} that we
restate below.

{\em   Let $f\in\PolRing$ and $\F\subset\PolRing$ satisfying assumptions
  $\assumptR$. There exists a non-empty Zariski open set
  $\NEZOS\subset\GLnC$ such that for all $\bfA\in\GLnQ\cap\NEZOS$, the
  properties $\PropReg{f^\bfA,\F^\bfA}$, $\PropPJSC{f^\bfA,\F^\bfA}$
  and $\PropPISSAC{f^\bfA,\F^\bfA}$ hold.
}

Actually, we prove that Property $\PropReg{f,\F}$ is always true if
$\F$ satisfies assumptions $\assumptR$. Next, we prove that there
exists a non-empty Zariski open set $\NEZOSJSCP\subset {\rm GL}_n(\C)$
such that for any $\bfA\NEZOSJSCP$, Property
$\PropPJSC{f^\bfA,\F^\bfA}$ holds. Likewise, we prove that there
exists a non-empty Zariski open set $\NEZOSISSAC\subset {\rm
  GL}_n(\C)$ such that for any $\bfA\in \NEZOSISSAC$, Property
$\PropPISSAC{f,\F}$ is satisfied. Then for any $\bfA$ in the non-empty
Zariski open set $\NEZOS = \NEZOSJSCP\cap\NEZOSISSAC$, the three
properties $\PropReg{f^\bfA,\F^\bfA}$, $\PropPJSC{f^\bfA,\F^\bfA}$ and
$\PropPISSAC{f^\bfA,\F^\bfA}$ hold.

The first two results are minor generalizations of \cite[Lemma
2.2]{GrGuSaZh} and \cite[Lemma 2.3]{GrGuSaZh}, where $V$ is assumed to
be smooth. The proofs of these lemmas in \cite{GrGuSaZh} can be
extended mutatis mutandis to our case by noticing that $x$ is a singular
point of $\Variety{\F, f-t}$ if and only it is a singular point of $V$
or a point such that $t=f\Par{x}$ is a critical value of $\fVC$.
\begin{proposition}\cite[Lemma 2.2]{GrGuSaZh}
  If $\F$ satisfies assumptions $\assumptR$ then $\PropReg{f,\F}$
  holds.
\end{proposition}

\begin{proposition}\cite[Lemma 2.3]{GrGuSaZh}
  There exists a non-empty Zariski open set $\NEZOSJSCP\subset\GLnC$
  such that for all $\bfA\in\GLnQ\cap\NEZOSJSCP$,
  $\PropPJSC{f^\bfA,\F^\bfA}$ holds.
\end{proposition}

Finally, we prove the following.
\begin{proposition}\label{prop:generic3}
  There exists a non-empty Zariski open set $\NEZOSISSAC\subset\GLnC$
  such that for all $\bfA\in\GLnQ\cap\NEZOSISSAC$,
  $\PropPISSAC{f^\bfA,\F^\bfA}$ holds.
\end{proposition}

We recall the first two points in \cite[Theorem 3, pp 134]{GrSa11}:
\begin{theorem}\label{thm:propernessISSAC}
  Let $V\subset\C^n$ be an algebraic variety of dimension $d$. There
  exists a non-empty Zariski open set $\NEZOSISSAC\subset\GLnC$ such
  that for all $\bfA\in\GLnQ\cap\NEZOSISSAC$, and $1\leq i \leq d+1$,
  there exist algebraic sets $V_{n-i+1}^\bfA\subset V^\bfA$ such that
  for all connected component $C^\bfA$ of $V^\bfA\cap\R^n$,
  \begin{itemize}
  \item[$(i)$] the restriction of $\ProjLeft{i-1}$ to $V^\bfA_{n-i+1}$
    is proper;
  \item[$(ii)$] the boundary of $\ProjLeft{i}\Par{C^\bfA}$ is
    contained in $\ProjLeft{i}\Par{C^\bfA\cap V^\bfA_{n-i+1}}$.
  \end{itemize}
\end{theorem}

Then we state some notations about infinitesimals and Puiseux
series. We denote by $\R\langle\varepsilon\rangle$ the real closed
field of algebraic Puiseux series with coefficients in $\R$, where
$\varepsilon$ is an infinitesimal. We use the classical notions of
bounded elements in $\R\langle\varepsilon\rangle^n$ over $\R^n$ and
their limits. The limit of a bounded element $z\in
\R\langle\varepsilon\rangle^n$ is denoted by $\limzero\Par{z}$. The
ring homomorphism $\limzero$ is also used on sets of
$\R\langle\varepsilon\rangle^n$. For semi-algebraic sets $S\subset
\R^n$ defined by a system of polynomial equations, we denote by
$\ext{S}$ the solution set of the considered system in $\R\langle
\varepsilon\rangle^n$. We refer to \cite[Chapter 2.6]{BaPoRo06} for
precise statements of these notions.  We can now prove Proposition
\ref{prop:generic3}.
\begin{proof}
  Let $\bfA\in\GLnQ\cap\NEZOSISSAC$ and $c$ be a critical value of
  $\fVbfA$ not isolated in $f^\bfA\Par{V^\bfA\cap\R^n}$. We prove that
  there exists $x_c\in\VPCtot{f^\bfA}{\F^\bfA}\cap\R^n$ such that
  $f^\bfA\Par{x_c}=c$. Let $\CCbfA$ be a connected component of
  $\Variety{f^\bfA-c}\cap V^\bfA\cap\R^n$.

  Consider the largest $i\in\left\{1,\dots,d\right\}$ such that
  $\CCbfA\cap\LinSpace{i-1}\neq\emptyset$ while
  $\CCbfA\cap\LinSpace{i}=\emptyset$.
  
  Let $\ProjISSAC{i}$ be the projection
  $$\fonction{\ProjISSAC{i}}{\C^n}{\C}{(x_1,\ldots,x_n)}{x_i}.$$ Then
  $\ProjISSAC{i}\Par{\CCbfA\cap\LinSpace{i-1}}\subset\R-\{0\}$ is a strict
  subset of $\R$. Moreover, it is closed because of $(i)$ and $(ii)$
  in Theorem \ref{thm:propernessISSAC}. Then every extremum of the
  projection $\ProjISSAC{i}$ is reached. Since
  $\ProjISSAC{i}\Par{\CCbfA\cap\LinSpace{i-1}}\neq\R$, there exists at
  least either a minimizer or a maximizer of $\ProjISSAC{i}$. Without
  loss of generality, we assume that it is a local minimizer, denoted
  by $x^\star$.

  Since $c$ is not an isolated point in $f^\bfA\Par{V^\bfA\cap\R^n}$,
  the set
  \[\Par{\Variety{f^\bfA-c-\eps}\cup\Variety{f^\bfA-c+\eps}}\cap
  V^\bfA\cap\LinSpace{i-1}\cap\R^n\] is nonempty. Then by \cite[Lemma
  3.6]{rrs}, the following sets coincide:
  \begin{itemize}
  \item $\Variety{f^\bfA-c}\cap V^\bfA\cap\LinSpace{i-1}\cap\R^n$
  \item $\limzero\Par{\Variety{f^\bfA-c\pm\eps}\cap V^\bfA\cap\LinSpace{i-1}}\cap\R^n$
  \end{itemize}
  Then, there exists a connected component
  $\CCbfA_\eps\subset\Rpuiseux^n$ of
  \[\Variety{f^\bfA-c\pm\eps}\cap V^\bfA\cap\LinSpace{i-1}\cap\Rpuiseux^n\]
  such that $\CCbfA_\eps$ contains a point $x_\eps$ such that
  $\limzero\Par{x_\eps}=x^\star$. Moreover, we can assume that
  $x_\eps$ minimize the projection $\ProjISSAC{i}$ over
  $\CCbfA_\eps$. Indeed, in the converse, there exists
  $x_\eps'\in\CCbfA_\eps$ such that
  $\ProjISSAC{i}\Par{x_\eps'}<\ProjISSAC{i}\Par{x_\eps}$, that implies
  $\limzero\ProjISSAC{i}\Par{x_\eps'}\leq
  \ProjISSAC{i}\Par{x^\star}$. Since $x^\star$ is a minimizer, this
  implies that $\limzero\ProjISSAC{i}\Par{x_\eps'} =
  \ProjISSAC{i}\Par{x^\star}$ and we replace $x_\eps$ with $x_\eps'$
  (note that $\lim_0(x_\eps')$ is not necessarily $x^\star$ but its
  image by $\ProjISSAC{i}$ is the same as $\ProjISSAC{i}(x^\star)$).

  As a minimizer of the projection, $x_\eps$ lies in the algebraic set
  defined as the vanishing set of
  \begin{itemize}
  \item the polynomials in $\F^\bfA$,
  \item the minors of size $n-d+1$ of $\jac{\left[f^\bfA-c\pm\eps,\F^\bfA\right]}{i+1}$,
  \item and the variables $X_1,\dots,X_{i-1}$.
  \end{itemize}
  Since
  $\jac{\left[f^\bfA-c\pm\eps,\F^\bfA\right]}{i+1}=\jac{\left[f^\bfA,\F^\bfA\right]}{i+1}$,
  this algebraic set is exactly
  $\ext{\VVA{f^\bfA}{\F^\bfA}{i}}$. Furthermore, since $\eps$ is an
  infinitesimal, $c\pm\eps$ is not a critical value of $f^\bfA$. Then
  $x_\eps\not\in\ext{\critfVbfA}$. This means that $x^\star$ is the
  limit of a sequence that lies in
  $\ext{\ZariskiClosure{\VVA{f^\bfA}{\F^\bfA}{i}\setminus
      \critfVbfA}}$. Hence $x^\star = \limzero{x_\eps}$ lies in
  $\ZariskiClosure{\VVA{f^\bfA}{\F^\bfA}{i}\setminus
    \critfVbfA}$. Moreover since $f^\bfA\Par{x^\star}=c$ that is a
  local extremum of $\fVbfA$, $x^\star\in\critfVbfA$. In other words,
  \[x^\star\in\ZariskiClosure{\VVA{f^\bfA}{\F^\bfA}{i}\setminus \critfVbfA}\cap\critfVbfA =
  \VPC{f^\bfA}{\F^\bfA}{i},\] that concludes the proof.
\end{proof}

Finally, Theorem \ref{thm:generic} is true with $\NEZOS =
\NEZOSJSCP\cap\NEZOSISSAC$. Since $\NEZOSJSCP$ and $\NEZOSISSAC$ are
non-empty Zariski open sets, $\NEZOS$ is also a non-empty Zariski open
set.

\section{Complexity analysis}\label{sec:complexity}
\subsection{Geometric degree bounds}\label{sec:degbounds}
In this section, we assume that the polynomial $f$ and the polynomials
$f_i$ have degree $\leq D$.  Recall that the degree of an irreducible
algebraic variety $V\subset\C^n$ is defined as the maximum finite
cardinal of $V\cap L$ for every linear subspace $L\subset\C^n$. If $V$
is not irreducible, $\deg V $ is defined as the sum of the degrees of
its irreducible components. The degree of a hypersurface $\Variety{f}$
is bounded by $\deg f$. Given a variety $V=\Variety{g_1,\dots,g_p}$,
we denote by $\delta\Par{V}$ the maximum of the degrees
$\deg\Par{V\Par{g_1,\dots,g_i}}$, for $1\leq i\leq p$.  Our goal in
this section is to estimate the degree of the geometric objects
computed in the algorithm. Obtaining such estimations is relevant
since the complexity of the algorithm relies on these degrees.

\begin{proposition}\label{prop:degbounds}
  For all $\bfA\in\GLnQ\cap\NEZOS$, for $1\leq i\leq d$,
  $\delta\Par{\VVA{f^\bfA}{\F^\bfA}{i}}$ and
  $\delta\Par{\VPC{f^\bfA}{\F^\bfA}{i}}$ are bounded by
$D^{n-d}\left(\left(n-d+1\right)\left(D-1\right)\right)^{d+1}$.
\end{proposition}
\begin{proof}
  Let $E_1=\Variety{f^\bfA-T, \F^\bfA}$ and denote by $E_2,E_3,\dots,
  E_p$ the zero-sets of each minor of size $n-d+1$ of
  $\jac{\left[f^\bfA,\F^\bfA\right]}{i+1}$. Then for $2\leq j\leq p$,
  each $E_j$ has degree bounded by $\Par{n-d+1}\Par{D-1}$. Moreover,
  it is straightforward to see that $E_1$ has degree bounded by
  $D^{n-d}$ and dimension $d$. Let $1\leq k\leq p$. Then using
  \cite[Proposition 2.3]{HeSc80} we get
  \begin{equation}
    \deg\Par{\bigcap_{1\leq j\leq k} E_j}\leq\deg E_1\Par{\max_{1<j\leq k}\deg E_j}^{\dim E_1}. \label{eq:HeSc}
  \end{equation}

  In particular,
  \[\deg\Par{\bigcap_{1\leq j\leq k} E_j}\leq D^{n-d}\Par{\Par{n-d+1}\Par{D-1}}^{d}.\]
  By B\'ezout's inequality (\cite[Proposition 2.3]{HeSc80}), it
  follows that the degree of $\displaystyle\bigcap_{1\leq j\leq k} E_j\cap
  \Variety{\bfX_{\leq i-1}}$ is also bounded by
  $D^{n-d}\Par{\Par{n-d+1}\Par{D-1}}^{d}$. Finally, this implies that
  \begin{equation}
    \delta\Par{\VVA{f^\bfA}{\F^\bfA}{i}}\leq
    D^{n-d}\left(\left(n-d+1\right)\left(D-1\right)\right)^{d}.\label{eq:geomdegbound}
  \end{equation}

  It remains to prove that $\delta\Par{\VPC{f^\bfA}{\F^\bfA}{i}}\leq
  D^{n-d}\left(\left(n-d+1\right)\left(D-1\right)\right)^{d+1}$. From the
  above inequality \eqref{eq:geomdegbound}, we deduce that
  \[\delta\Par{\ZariskiClosure{\VVA{f^\bfA}{\F^\bfA}{i}\setminus{\critfVbfA}}}\leq
  D^{n-d}\left(\left(n-d+1\right)\left(D-1\right)\right)^{d}.\]

  Finally, we apply \cite[Proposition 2.3]{HeSc80} with the varieties
  $F_1,\dots,F_t$, where
  \[F_1=
  \ZariskiClosure{\VVA{f^\bfA}{\F^\bfA}{i}\setminus{\critfVbfA}}\] and
  $F_2,F_3,\dots,F_t$ are the zero-sets of each minor defining
  $\critfVbfA$.  Since these minors have degree bounded by
  $\Par{n-d+1}\Par{D-1}$, this quantity bounds the degree of the
  hypersurfaces they define. By Proposition \ref{prop:finite}, $F_1 =
  \ZariskiClosure{\VVA{f^\bfA}{\F^\bfA}{i}\setminus\critfVbfA}$ has
  dimension $1$. Then inequality \eqref{eq:HeSc} becomes
  \[\deg\Par{\bigcap_{1\leq j\leq t} F_j}\leq
  D^{n-d}\left(\left(n-d+1\right)\left(D-1\right)\right)^{d}\times \Par{n-d+1}\Par{D-1}.\] We conclude  that
  \[\delta\Par{\VPC{f^\bfA}{\F^\bfA}{i}}\leq
  D^{n-d}\left(\left(n-d+1\right)\left(D-1\right)\right)^{d+1}.\]
\end{proof}

\subsection{Complexity estimates}
Let $\bfA\in\GLnQ\cap\NEZOS$. Let
$\mathbf{F}=\left\{f_1,\dots,f_\nV\right\}\subset\PolRing$, $f$ and
$g$ in $\PolRing$ of degree bounded by $D$ that are given by a
straight-line program (SLP) of size $\leq L$. Recall that $d$ denotes
the dimension of $V=\Variety{\mathbf{F}}$.

To estimate the complexity of our algorithm, we use the the geometric
resolution algorithm \cite{GiLeSa01, lecerf2003} and its
subroutines. This is a probabilistic algorithm for polynomial system
solving. In the sequel, a geometric resolution is a representation of
a variety by a parametrization. A lifting fiber is a data from which a
geometric resolution can be recovered. We refer to
\cite{GiLeSa01,lecerf2003, schost03} for precise statements.

We describe the geometric resolution probabilistic subroutines and
their complexity used to represent the varieties in our algorithm in
Section \ref{sec:geomressubroutines}. The complexity depends on the
geometric degree and the maximal size of the SLPs representing the
polynomials involved in the definition of our varieties. Then Section
\ref{sec:sizeSLP} is devoted to give bounds on the size of these
SLPs. Then estimations of the complexity of our subroutines and the
main algorithm are given in the rest of this Section.

\subsubsection{Geometric Resolution subroutines}\label{sec:geomressubroutines}
\begin{itemize}
\item \textsf{GeometricSolveRRS} \cite[Section 7]{GiLeSa01}: let
  $\F=\left\{f_1,\dots,f_n\right\}$ and $g$ be polynomials in
  $\Q[\mathbf{X}]$ of degree $\leq D$ and given by a straight-line
  program of length $E$. Assume that $\F$ defines a {\em reduced
    regular sequence} in the open subset $\left\{g\neq
    0\right\}$. This routine returns a geometric resolution of
  $\ZariskiClosure{\Variety{\mathbf{F}}\setminus \Variety{g}}$ in
  probabilistic time
  \[\widetilde{O}\Par{\Par{nE+n^4}\Par{D\delta\Par{\Variety{\F}}}^2}.\]
\item \textsf{GeometricSolve} \cite[Section 5.2]{lecerf2003}: let
  $\mathbf{F}$ and $g$ be as above of degree $\leq D$ given by a straight-line program of
  length $E$. This routine returns an equidimensional decomposition of
  $\ZariskiClosure{\Variety{\mathbf{F}}\setminus \Variety{g}}$,
  encoded by a set of irreducible lifting fibers in probabilistic time
  \[\widetilde{O}\Par{sn^4\Par{nE+n^4}\Par{D\delta\Par{\Variety{\F}}}^3}.\]
\item \textsf{SplitFiber} \cite[Section 3.4]{lecerf2003}: given a
  lifting fiber $F$ of a variety $\Variety{\F}$, it returns a set of
  irreducible lifting fibers of $\Variety{\F}$ in time 
  \[\widetilde{O}\Par{sn^4\Par{nE+n^4}\Par{D\delta\Par{\Variety{\F}}}^3}.\]
\item \textsf{LiftCurve} \cite[Section 3.3]{lecerf2003}: given an irreducible
  lifting fiber $F$ of the above output,  this
  routine returns a rational parametrization of the lifted curve of
  $F$ in time
  \[\widetilde{O}\Par{sn^4\Par{nE+n^4}\Par{D\delta\Par{\Variety{\F}}}^2}.\]
\item \textsf{OneDimensionalIntersect} \cite[Section 6]{GiLeSa01}: let
  $\F$ be as above such that $\IdealAngle{\F}$ is a $1$-dimensional ideal,
  $\mathfrak{I}$ be a geometric resolution of $\IdealAngle{\F}$, and
  $f$ and $g$ be polynomials. In case of success, the routine returns
  a rational parametrization of
  $\ZariskiClosure{\Variety{\mathfrak{I}+f}\cap\Variety{g}}$ in time
  \[\widetilde{O}\Par{n\Par{E+n^2}\Par{D\delta\Par{\Variety{\F}}^2}}.\]
\item \textsf{LiftParameter} \cite[Section 4.2]{schost03}: let $T$ be
  a parameter and let $\mathcal{P}_T$ be a set of polynomials in
  $\Q(T)[X_1,\dots,X_n]$ be given by a straight-line program of length
  $E$. Let $t\in\R$ be a generic point and $\mathcal{P}_t$ be the
  polynomial system specialized at $t$. If $\Variety{\mathcal{P}_t}$
  is $0$-dimensional, the routine takes as input a geometric
  resolution of $\mathcal{P}_t$ and returns a parametric geometric
  resolution of $\mathcal{P}_t$ in time
  \[\widetilde{O}\Par{\Par{nE+n^4+n}\delta\Par{\Variety{\mathcal{P}_t}}
    \Par{4\delta\Par{\Variety{\mathcal{P}_T}}+1}}.\]
\item \textsf{Difference} \cite[Section 4.1]{lecerf2003}: let $V_1$,
  $V_2$ be algebraic varieties defined as the vanishing set of
  polynomial families given by a straight line program of length $E$
  and represented by lifting fibers. In case of success, the routine
  returns a fiber $F$ of the components of
  $\ZariskiClosure{V_1\setminus V_2}$ in time
  \[\widetilde{O}\Par{n^4\Par{nE+n^4}\delta\Par{V_2}^2\delta\Par{V_1}}.\]
\end{itemize}

\subsubsection{Size of SLP}\label{sec:sizeSLP}
We want to estimate some parameters depending on the inputs of the
geometric resolution routines, that are the polynomials defining the
varieties $\VVA{f^\bfA}{\F^\bfA}{i}$ and
$\VPC{f^\bfA}{\F^\bfA}{i}$. Since bounds on
$\delta\Par{\VVA{f^\bfA}{\F^\bfA}{i}}$ and
$\delta\Par{\VPC{f^\bfA}{\F^\bfA}{i}}$ have been obtained in the
previous section, it remains to estimates the size of the
straight-line programs representing these polynomials. These
polynomials are either a polynomial $f^\bfA$ or $f^\bfA_i$ or a minor
of size $n-d+1$ of the Jacobian matrix
$\jac{\left[f^\bfA,\F^\bfA\right]}{i+1}$. The polynomials $f$ and
$f_i$ are given as a SLP of size $L$. Then $f^\bfA$ and $f^\bfA_i$,
can be represented by a SLP of size $O\left(L + n^2\right)$. Then we
estimate the size of the minors. Let $\omega$ be the
matrix-multiplication exponent.

\begin{proposition}
  Each minors of size $n-d+1$ of
  $\jac{\left[f^\bfA,\F^\bfA\right]}{i+1}$ can be represented by a SLP
  of size
  $\widetilde{O}\left(\left(n-d+1\right)^{\omega/2+2}\left(L+n^2\right)\right)$.
\end{proposition}
\begin{proof}
  The partial derivatives appearing in the Jacobian matrix come from
  $f^\bfA$ and $f^\bfA_i$, represented by a SLP of size
  $O\left(L+n^2\right)$. According to \cite{BaurStrassen}, each
  partial derivative $\frac{\partial f^\bfA_i}{\partial x_j}$ and
  $\frac{\partial f^\bfA}{\partial x_j}$ can be represented by a SLP
  of size $O\left(L+n^2\right)$. Moreover, according to \cite{Ka92},
  the determinant of an $n\times n$ matrix can be computed using only
  $+$, $-$ and $\times$ in
  $\widetilde{O}\left(\left(n-d+1\right)^{\omega/2+2}\right)$
  operations. We combine these two results to conclude the proof.
\end{proof}

\begin{remark}
  Recall that $\omega\leq 3$. In the sequel, to lighten the
  expressions of complexity, we replace the above complexity
  $\widetilde{O}\left(\left(n-d+1\right)^{\omega/2+2}\left(L+n^2\right)\right)$
  with $\widetilde{O}\left(n^{4}\left(L+n^2\right)\right)$, that is
  less accurate but that dominates the first one.
\end{remark}

\subsubsection{Computation of  $\ZariskiClosure{\VVA{f^\bfA}{\F^\bfA}{i}\setminus\crit{f^\bfA,V^\bfA}}$}\label{sec:complexVVA}
Recall that the algebraic variety $\VVA{f^\bfA}{F^\bfA}{i}$ is defined as the vanishing set
of
\begin{itemize}
\item the polynomials $f^\bfA_1, \ldots, f^\bfA_\nV$,
\item the minors of size $n-d+1$ of $\jacc{\left[f^\bfA,\F^\bfA\right], i+1}$,
\item and the variables $X_1, \ldots, X_{i-1}$.
\end{itemize}

The algebraic set $\VVA{f^\bfA}{F^\bfA}{i}$ can be computed by
\textsf{GeometricSolve} called with $\nV +
\binom{\nV+1}{n-d+1}\binom{n-i}{n-d+1} = O\!\left(\nV +
  \binom{\nV+1}{n-d+1}\binom{n}{n-d+1}\!\right)$ polynomials in $n$
variables. Each polynomial is given by a SLP of size
$\widetilde{O}\!\left(n^4\!\left(L+n^2\right)\right)$. Hence, the
input system is represented by a SLP of size $E$ in
$\widetilde{O}\left(\left(\nV +
    \binom{\nV+1}{n-d+1}\binom{n}{n-d+1}\!\right)
  n^4\!\left(L+n^2\right) \right)$. By Proposition
\ref{prop:degbounds}, $\delta\Par{\VVA{f^\bfA}{F^\bfA}{i}}$ is bounded
by
$D^{n-d}\left(\left(n-d+1\right)\left(D-1\right)\right)^{d+1}$. Hence,
the computation can be done within
\[\widetilde{O}\left(s\Par{\nV+\binom{\nV+1}{n-d+1}\binom{n}{n-d+1}}
  LD^{3(n-d+1)}\left(\left(n-d+1\right)\left(D-1\right)\right)^{3(d+1)}\right)\]
arithmetic operations in $\Q$. Since $\binom{n}{n-d+1}\leq 2^n$, this
complexity is bounded by $\widetilde{O}\left(\left(s^2+s2^n\binom{\nV+1}{n-d+1}
  \right)LD^{3(n-d+1)}\left(\left(n-d+1\right)\left(D-1\right)\right)^{3(d+1)}\right)$.

Likewise, $\crit{f^\bfA,\F^\bfA}$ is defined as the vanishing set of
\begin{itemize}
\item the polynomials $f^\bfA_1, \ldots, f^\bfA_\nV$,
\item the minors of size $n-d+1$ of $\jacc{\left[f^\bfA,\F^\bfA\right]}$.
\end{itemize}
Hence, the complexity of its computation by \textsf{GeometricSolve} is
the same as the above complexity of the computation of
$\VVA{f^\bfA}{F^\bfA}{i}$.

The computation of
$\ZariskiClosure{\VVA{f^\bfA}{\F^\bfA}{i}\setminus\crit{f^\bfA,V^\bfA}}$
is done using \textsf{Difference}. Its complexity is in
$\widetilde{O}\left(ED^{3(n-d)}\left(\left(n-d+1\right)\left(D-1\right)\right)^{(d+1)}\right)$.
It is dominated by the cost of the computation of
$\VVA{f^\bfA}{\F^\bfA}{i}$, thus we get the following complexity
result.
\begin{lemma}\label{lemma:complmodvp}
  There exists a probabilistic algorithm that takes as input
  $f^\bfA,\F^\bfA$ and $i$ and returns an equidimensional
  decomposition of
  $\ZariskiClosure{\VVA{f^\bfA}{F^\bfA}{i}\setminus\crit{f^\bfA,V^\bfA}}$,
  encoded by a lifting fiber. In case of success,
  the algorithm has a complexity dominated by
  $\widetilde{O}\left(\left(s^2+s2^n\binom{\nV+1}{n-d+1}\right)
  LD^{3(n-d+1)}\left(\left(n-d+1\right)\left(D-1\right)\right)^{3(d+1)}\right)$.
\end{lemma}

\subsubsection{Computation of $\VPC{f^\bfA}{\F^\bfA}{i}$}
Since $\VPC{f^\bfA}{\F^\bfA}{i}$ is defined as the intersection
$\displaystyle\ZariskiClosure{\VVA{f^\bfA}{\F^\bfA}{i}\setminus\critfVbfA}\cap\critfVbfA$,
a geometric resolution of each component of $\VPC{f^\bfA}{\F^\bfA}{i}$
is obtained from a set of irreducible lifting fibers of
$\ZariskiClosure{\VVA{f^\bfA}{\F^\bfA}{i}\setminus\critfVbfA}$. From
the output of the algorithm presented in Section \ref{sec:complexVVA},
a set of irreducible lifting fibers is recovered using
\textsf{SplitFiber} in time $\widetilde{O}\left(sE
  D^{3(n-d+1)}\left(\left(n-d+1\right)\left(D-1\right)\right)^{3(d+1)}\right)$. As
in Section \ref{sec:complexVVA}, the size $E$ is in
$\widetilde{O}\left(\left(\nV +
    \binom{\nV+1}{n-d+1}\binom{n}{n-d+1}\!\right)
  n^4\!\left(L+n^2\right) \right)$.
The routine \textsf{LiftCurve} is used on each irreducible fiber in
order to obtain a parametrization of each component of the curve
$\ZariskiClosure{\VVA{f^\bfA}{\F^\bfA}{i}\setminus\critfVbfA}$. Lifting
one fiber is done in
\[\widetilde{O}\left(s\Par{\nV+\binom{\nV+1}{n-d+1}\binom{n}{n-d+1}}
  L D^{2(n-d+1)}\left(\left(n-d+1\right)\left(D-1\right)\right)^{2(d+1)}\right).\]

From such a parametrization, the routine
\textsf{OneDimensionalIntersect} is used with every polynomial that
defines $\critfVbfA$. There are $\binom{\nV+1}{n-d+1}\binom{n}{n-d+1}$
such polynomials, thus the cost to compute the rational
parametrization corresponding to one irreducible fiber is at
most \[\widetilde{O}\left(\binom{\nV+1}{n-d+1}\binom{n}{n-d+1} L
  D^{2(n-d+1)}\left(\left(n-d+1\right)\left(D-1\right)\right)^{2(d+1)}\right).\]
The total cost for the
$D^{n-d}\left(\left(n-d+1\right)\left(D-1\right)\right)^{d+1}$
irreducible lifting fibers is dominated by
$\widetilde{O}\left(s\Par{\nV+\binom{\nV+1}{n-d+1}\binom{n}{n-d+1}}
  LD^{3(n-d+1)}\left(\left(n-d+1\right)\left(D-1\right)\right)^{3(d+1)}\right)$.
Finally, we compute a parametrization of the union of the zero-sets of
the previous parametrization using \cite[Lemma
9.1.3]{din2013nearly}. Since the sum of the degrees of each
parametrization is bounded by
$D^{n-d}\left(\left(n-d+1\right)\left(D-1\right)\right)^{d+1}$, the
cost is negligible. Since $\binom{n}{n-d+1}\leq 2^n$, we get the
following result.

\begin{lemma}\label{lemma:complVPC}
  There exists a probabilistic algorithm that takes as input a set of
  lifting fibers of
  $\ZariskiClosure{\VVA{f^\bfA}{\F^\bfA}{i}\setminus\critfVbfA}$ and
  that returns a rational parametrization of
  $\VPC{f^\bfA}{\F^\bfA}{i}$. In case of success, the algorithm has a
  complexity dominated by
  $\widetilde{O}\left(\left(s^2+s2^n\binom{\nV+1}{n-d+1}\right)
    LD^{3(n-d+1)}\left(\left(n-d+1\right)\left(D-1\right)\right)^{3(d+1)}\right)$.
\end{lemma}

\subsubsection{Complexity of \SetOfNonProperness{}}
As explained in \cite[Section 4]{Safey07}, the computation of the set of
non-properness of the restriction of the projection $\pi_T$ to
$\Variety{f^\bfA-T}\cap\left(\ZariskiClosure{\VVA{f^\bfA}{\F^\bfA}{i}\setminus\crit{f^\bfA,V^\bfA}}\times\C\right)$
from the representation of the variety
$\ZariskiClosure{\VVA{f^\bfA}{\F^\bfA}{i}\setminus\crit{f^\bfA,V^\bfA}}$
can be done using a parametric geometric resolution
\cite{schost03}. Indeed, from a set of lifting fibers of
$\ZariskiClosure{\VVA{f^\bfA}{\F^\bfA}{i}\setminus\crit{f^\bfA,V^\bfA}}$,
obtained by the routine \textsf{GeometricSolve}, one can compute a
geometric resolution of
$\Variety{f^\bfA-t}\cap\ZariskiClosure{\VVA{f^\bfA}{\F^\bfA}{i}\setminus\crit{f^\bfA,V^\bfA}}$
for a generic $t\in\R$. Since there are at most
$D^{n-d}\left(\left(n-d+1\right)\left(D-1\right)\right)^{d+1}$ lifting
fibers, it is done using \textsf{OneDimensionalIntersect} on all the fibers
in $\widetilde{O}\Par{E
  D^{3(n-d+1)}\left(\left(n-d+1\right)\left(D-1\right)\right)^{3(d+1)}}$,
where as in Section \ref{sec:complexVVA}, $E$ is in
$\widetilde{O}\left(\left(\nV +
    \binom{\nV+1}{n-d+1}\binom{n}{n-d+1}\!\right) n^4\!\left(L+n^2\right)
\right)$.

From these geometric resolutions, \textsf{LiftParameter} computes a
parametric geometric resolution $(q,q_0,\dots,q_n)$ of
$\Variety{f^\bfA-T}\cap\left(\ZariskiClosure{\VVA{f^\bfA}{\F^\bfA}{i}\setminus\crit{f^\bfA,V^\bfA}}\times\C\right)$,
where $T$ is a parameter, in $\widetilde{O}\Par{E
  D^{3(n-d)}\left(\left(n-d+1\right)\left(D-1\right)\right)^{3(d+1)}}$.  As
explained in \cite[Section 4]{Safey07}, the set of non-properness is
contained in the roots of the least common multiple of the denominators of
the coefficients of the polynomial $q$. Since $\binom{n}{n-d+1}\leq 2^n$,
we get the following result.

\begin{lemma}\label{lemma:complSONP}
  There exists a probabilistic algorithm that takes as input a set of
  lifting fibers of
  $\ZariskiClosure{\VVA{f^\bfA}{\F^\bfA}{i}\setminus\crit{f^\bfA,V^\bfA}}$
  and that returns a polynomial whose set of roots contains the set of
  non-properness of the projection $\pi_T$ restricted to
  $\Variety{f^\bfA-T}\cap\left(\ZariskiClosure{\VVA{f^\bfA}{\F^\bfA}{i}\setminus\crit{f^\bfA,V^\bfA}}\times\C\right)$. In
  case of success, the algorithm has a complexity dominated by
  $\widetilde{O}\Par{\left(s^2+s2^n\binom{\nV+1}{n-d+1}\right) L
    D^{3(n-d+1)}\left(\left(n-d+1\right)\left(D-1\right)\right)^{3(d+1)}}$.
\end{lemma}

\subsubsection{Complexity of \SamplePoints{} and \IsEmpty{}}
Given $\F=\left\{f_1,\dots,f_\nV\right\}$, an algorithm computing a
set of real sample points of $\Variety{\F}\cap\R^n$ is given in
\cite{SaSc}. It relies on the computation of polar varieties. Using
techniques described in \cite{BaGiHeMb01, BaGiHeMb97, din2013nearly}
and \cite[Section 3]{BGHSS}, a local description of the polar
varieties as a complete intersection can be obtained.  Assume that
$\Variety{\F}$ is equidimensional of dimension $d$. Such a local
description depends on the choice of a minor of size $n-d$ of the
Jacobian matrix $\jacc{\F}$. Since there are
$\binom{\nV}{n-d}\binom{n}{n-d}$ minors of size $n-d$ in $\jacc{\F}$,
a full description of the polar varieties is obtained by computing the
$\binom{\nV}{n-d}\binom{n}{n-d}$ possible localizations. Each local
description is given by a reduced regular sequence involving $n-d$
polynomials in $\F$ and minors of degree bounded by
$(n-d+1)(D-1)$. Hence, the routine \textsf{GeometricSolveRRS} computes
one local description in $\widetilde{O}\left(L
  D^{2(n-d+1)}\left(\left(n-d+1\right)\left(D-1\right)\right)^{2(d+1)}
\right)$ so the cost for all localizations is in
$\widetilde{O}\left(\binom{\nV}{n-d}\binom{n}{n-d} L
  D^{2(n-d+1)}\left(\left(n-d+1\right)\left(D-1\right)\right)^{2(d+1)}\right)$.
Note that in \cite{SaSc}, the complexity is cubic instead of quadratic
in the geometric degree of the inputs because these localization
techniques are not used to estimate the complexity.

Since
$\binom{n}{n-d}\leq 2^n$, this leads to the following complexity
result.
\begin{lemma}\label{lemma:complsample}
  There exists a probabilistic algorithm that takes as input $\F$
  satisfying assumptions $\assumptR$ and that returns a set of real
  sample points of $\Variety{F}\cap\R^n$, encoded by a rational
  parametrization. In case of success, the algorithm has a complexity
  dominated by $\widetilde{O}\left(2^n\binom{\nV}{n-d}L
    D^{2(n-d+1)}\left(\left(n-d+1\right)\left(D-1\right)\right)^{2(d+1)}\right)$.
\end{lemma}

\subsubsection{Complexity of \SetContaining{}}
The first step in \SetContaining{} is the computation of a set of real
sample points of $\Variety{\F^\bfA}\cap\R^n$. Its complexity is given
in Lemma \ref{lemma:complsample}. Then at the $i$-th step of the loop,
$\ZariskiClosure{\VVA{f^\bfA}{\F^\bfA}{i}\setminus\crit{f^\bfA,V^\bfA}}$,
the set of non-properness of the projection $\pi_T$ restricted to
$\Variety{f-T}\cap\left(\ZariskiClosure{\VVA{f^\bfA}{\F^\bfA}{i}\setminus\crit{f^\bfA,V^\bfA}}\times\C\right)$
and $\VPC{f^\bfA}{\F^\bfA}{i}$ are computed. The costs are given in
Lemma \ref{lemma:complmodvp}, Lemma \ref{lemma:complSONP} and Lemma
\ref{lemma:complVPC}. The complexity for one step is in
$\widetilde{O}\left(\left(s^2 + s2^n\binom{\nV+1}{n-d+1}\right)
  L D^{3(n-d+1)}\left(\left(n-d+1\right)\left(D-1\right)\right)^{3(d+1)}\right)$. Finally,
for the $d$ steps, since $d\leq n$ can be omitted, we get the
following complexity.

\begin{lemma}\label{lemma:complSC}
  In case of success, the routine \SetContaining{} performs at most
  $\widetilde{O}\left( \left(s^2 + s2^n \binom{\nV+1}{n-d+1}\right) L
    D^{3(n-d+1)}\left(\left(n-d+1\right)\left(D-1\right)\right)^{3(d+1)}\right)$
  arithmetic operations in $\Q$.
\end{lemma}

\subsubsection{Complexity of \FindInfimum{}}
The most expensive steps in this routine are the calls to
\IsEmpty{}. There are at most $k$ such steps, where $k$ is the number
of points of non-properness, of critical values and of real sample
points. Using the Bézout inequality, $k$ lies in
$\widetilde{O}\Par{D^{n-d}\left(\left(n-d+1\right)\left(D-1\right)\right)^{d+1}}$. Using
the complexity estimate given in Lemma \ref{lemma:complsample}, this
leads to the following.

\begin{lemma}\label{lemma:complFI}
  In case of success, the routine \FindInfimum{} performs at most
  \[\widetilde{O}\left(2^n\binom{\nV}{n-d}L
    D^{3(n-d+1)}\left(\left(n-d+1\right)\left(D-1\right)\right)^{3(d+1)}\right).\]
\end{lemma}

\subsubsection{Complexity of the Algorithm}
Finally, the complexity of \Optimize{} comes from Lemma
\ref{lemma:complSC} and Lemma \ref{lemma:complFI}, using that
$\binom{s}{n-d}\leq \binom{s+1}{n-d+1}$.

\begin{theorem}\label{thm:complexity}
  In case of success, the algorithm \Optimize{} performs
\[\widetilde{O}\left(\left(s^2+s
  2^n \binom{\nV+1}{n-d+1}\right) L
  D^{3(n-d+1)}\left(\left(n-d+1\right)\left(D-1\right)\right)^{3(d+1)}\right)\]
arithmetic operations in $\Q$.
\end{theorem}

Remark that if $\F$ is a reduced regular sequence then the 
complexity is simpler.
\begin{theorem}
  If $\F$ is a reduced regular sequence, the algorithm \Optimize{}
  performs $\widetilde{O}\left(
    L\left(\sqrt[3]{2}\left(s+1\right)D\right)^{3(n+2)}\right)$
arithmetic operations in
  $\Q$.
\end{theorem}



\section{Implementation and practical experiments}\label{sec:practical}
We give details about our implementation in Section
\ref{sec:implem}. Instead of using the geometric resolution algorithm
for algebraic elimination, we use Gr\"obner bases that still allow to
perform all geometric operations needed to implement the algorithm
(see \cite{cox} for an introduction to Gr\"obner bases). Moreover,
there exist practically efficient algorithms for computing Gr\"obner
bases \cite{F4, F5}.  This way, the probabilistic aspect of our
algorithm relies on the random choice of a linear change of
variables. In practice, we check whether the linear change of
variables is suitable. Thus one can guarantee exactness. This is
explained in Section \ref{sec:implem}.

In Sections \ref{sec:benchrandom} and \ref{sec:benchapplications}, we
present practical experiments. First, we run our implementation with
random dense polynomials, that is the hardest case for the inputs.  As
an example, our implementation can solve problems with an objective
polynomial and one constraint, both of degree $2$, with up to $32$
variables using $4$ hours of CPU time. With two constraints, our
implementation can solve problems with up to $11$ variables using
$5.3$ hours of CPU time. With a linear objective polynomial subject to
one constraint of degree $4$, both in $5$ variables, it takes $34$
minutes. These results show that our implementation outperforms
general symbolic solvers based on the Cylindrical Algebraic
Decomposition.

Then we run examples coming from applications. Some of these examples
can be solved by QEPCAD. The timings are given in Section
\ref{sec:benchapplications}.

We do not report timings of methods based on sums of squares or
numerical procedures, e.g. \cite{YALMIP, SOSTOOLS, GloptiPoly} since
their outputs are numerical approximation while our algorithm provides
exact outputs.

\subsection{Implementation}\label{sec:implem}
Since our algorithm depends on the choice of a matrix that defines a
change of coordinates, it is probabilistic. However, we present a
technique to make sure that this choice is a correct one. This
technique is used in our implementation.

As stated in Section \ref{sec:correctness}, the algorithm is correct
if the subroutines \SetContaining{} and \FindInfimum{} are
correct. According to Proposition \ref{prop:SetContainingcorrect}, if
the random matrix chosen at the first step of \Optimize{} is such that
$\PropPJSC{f^\bfA,\F^\bfA}$, $\PropPISSAC{f^\bfA,\F^\bfA}$ and
$\PropReg{f^\bfA,\F^\bfA}$ hold, then \SetContaining{} runs
correctly. Then its output satisfies property $\PropOpt{W(f, \F),
  V(\F)}$. Hence, \FindInfimum{} can be called with the output of
\SetContaining{}.

Then the choice of the matrix $\bfA$ leads to a correct output if
$\PropPJSC{f^\bfA,\F^\bfA}$, $\PropPISSAC{f^\bfA,\F^\bfA}$ and
$\PropReg{f^\bfA,\F^\bfA}$ hold.

Property $\PropReg{f,\F}$ always holds if $\F$ satisfies assumptions
$\assumptR$ (see \cite[Lemma 2.2]{GrGuSaZh}). Since for any change of
coordinates, $\F$ satisfies assumptions $\assumptR$ if and only if
$\F^\bfA$ does, $\PropReg{f^\bfA,\F^\bfA}$ holds for any
$\bfA\in\GLnQ$. Then it remains to check $\PropPJSC{f^\bfA,\F^\bfA}$
and $\PropPISSAC{f^\bfA,\F^\bfA}$. Both properties depend on the
properness of projections of the form
\[\fonction{\ProjLeft{d}}{W\subset\C^n}{\C^d}{(x_1, \ldots, x_n)}{\left(x_1,\dots, x_d\right)}\]
where $W$ is an algebraic variety. According to \cite[Proposition
3.2]{jelonek}, if $I_V$ is an ideal such that $V=\Variety{I_V}$ has
dimension $d$ then the projection
\[\fonction{\ProjLeft{d}}{V\subset\C^n}{\C^d}{(x_1, \ldots, x_n)}{\left(x_1,\dots, x_d\right)}\]
is proper if and only if $I_V$ is in Noether position.

Thus we choose the matrix $\bfA$ such that after the change of
variables, the ideals are in Noether position. This can be done using
techniques described in \cite[Section 4.1.2]{KrPaSo01} and
\cite{logar89}. These techniques are used in our implementation to
obtain a matrix as sparse as possible that makes \SetContaining{}
correct.

\subsection{Practical experiments}\label{sec:benchrandom}
The analysis of the degree of the algebraic varieties involved in the
computations provides a singly exponential bound in the number of
indeterminates. This matches the best complexity bounds for global
optimization algorithms using quantifier elimination.  Our
implementation is written in Maple. Gr\"obner bases are computed using
the package FGb \cite{fgb} (\url{http://www-polsys.lip6.fr/~jcf/Software/}.)

The computations were performed on a Intel Xeon CPU E7540 @ 2.00GHz
and 250GB of RAM.

In the tables below, we use the following notations:
\begin{itemize}
\item $d$: degree of the objective polynomial $f$;
\item $D$: upper bound for the degree of the constraints;
\item $n$: number of indeterminates;
\item $\nV$: number of constraints;
\item obj terms: number of terms in the objective polynomial;
\item terms: average number of terms.
\end{itemize}

To test the behavior of the algorithm, we run it with randomly
generated polynomials and constraints as inputs.

Considering an objective polynomial and one constraint, both of degree
$2$ and increasing the number of variables, our implementation can
solve problems with up to $32$ variables in $4$ hours. For this
special case, the algorithm seems to be sub-exponential.

\paragraph{Constraints of degree 2}~
{\small\begin{center}
  \begin{tabular}{|ccccccc|}\hline
    $n$ & $d$ & $D$ & $s$ & obj terms & terms & time\\
    \hline
    8 & 2 & 2 & 1 & 44 & 45 & 9 sec.\\
    12 & 2 & 2 & 1 & 91 & 91 & 30 sec.\\
    16 & 2 & 2 & 1 & 153 & 153 & 2 min..\\
    20 & 2 & 2 & 1 & 229 & 231 & 8 min.\\
    24 & 2 & 2 & 1 & 323 & 323 & 27 min.\\
    28 & 2 & 2 & 1 & 433 & 433 & 1.5 hours\\
    32 & 2 & 2 & 1 & 559 & 557 & 4 hours\\
    7 & 2 & 2 & 2 & 36 & 36 & 92 sec.\\
    8 & 2 & 2 & 2 & 45 & 45 & 7 min.\\
    9 & 2 & 2 & 2 & 55 & 55 & 27 min.\\
    10 & 2 & 2 & 2 & 65 & 66 & 1.6 hours\\
    11 & 2 & 2 & 2 & 78 & 78 & 5.3 hours\\
    \hline
  \end{tabular}
\end{center}}


\paragraph{Constraints of degree 3}~
{\small \begin{center}
  \begin{tabular}{|ccccccc|}\hline
    $n$ & $d$ & $D$ & $s$ & obj terms & terms & time\\
    \hline
    4 & 2 & 3 & 1 & 15 & 34 & 4 sec.\\
    5 & 2 & 3 & 1 & 21 & 55 & 28 sec.\\
    6 & 2 & 3 & 1 & 27 & 84 & 9 min.\\
    7 & 2 & 3 & 1 & 36 & 120 & 3.5 hours\\
    4 & 2 & 3 & 2 & 15 & 34 & 81 sec.\\
    5 & 2 & 3 & 2 & 21 & 56 & 2.2 hours\\
    \hline
  \end{tabular}
\end{center}}

\paragraph{Constraints of degree 4}~
{\small \begin{center}
  \begin{tabular}{|ccccccc|}\hline
    $n$ & $d$ & $D$ & $s$ & obj terms & terms & time\\
    \hline
    2 & 3 & 4 & 1 & 10 & 14 & 2 sec.\\
    3 & 3 & 4 & 1 & 20 & 34 & 4 sec.\\
    4 & 3 & 4 & 1 & 34 & 70 & 7 min.\\
    3 & 3 & 4 & 2 & 20 & 35 & 22 sec.\\
    4 & 3 & 4 & 2 & 35 & 70 & 4.8 hours.\\
    2 & 2 & 4 & 1 & 6 & 15 & 1 sec.\\
    3 & 2 & 4 & 1 & 10 & 35 & 2 sec.\\
    4 & 2 & 4 & 1 & 15 & 68 & 83 sec.\\
    \hline
  \end{tabular}
\end{center}}

\paragraph{Linear objective function}~
{\small \begin{center}
  \begin{tabular}{|ccccccc|}\hline
    $n$ & $d$ & $D$ & $s$ & obj terms & terms & time\\
    \hline
    4 & 1 & 3 & 1 & 5 & 34 & 3 sec.\\
    4 & 1 & 4 & 1 & 5 & 69 & 30 sec.\\
    4 & 1 & 5 & 1 & 5 & 126 & 13 min.\\
    5 & 1 & 3 & 1 & 6 & 56 & 7 sec.\\
    5 & 1 & 4 & 1 & 6 & 126 & 34 min.\\
    5 & 1 & 5 & 1 & 6 & 252 & 87 hours\\
    6 & 1 & 3 & 1 & 7 & 84 & 68 sec.\\
    6 & 1 & 4 & 1 & 7 & 207 & 62 hours\\
    4 & 1 & 3 & 2 & 5 & 35 & 36 sec.\\
    4 & 1 & 4 & 2 & 5 & 70 & 1 hour\\
    4 & 1 & 5 & 2 & 5 & 126 & 33 hours\\
    \hline
  \end{tabular}
\end{center}}

\subsection{Examples coming from applications}\label{sec:benchapplications}
We consider examples coming from applications to compare the execution
time of our algorithm with a cylindrical algebraic decomposition
algorithm. These decompositions are computed using QEPCAD version B
1.69\footnote{Implementation originally due to H. Hong, and
  subsequently added on to by C. W. Brown, G. E. Collins,
  M. J. Encarnacion, J. R. Johnson, W. Krandick, S. McCallum,
  S. Steinberg, R. Liska, N. Robidoux. Latest version is available at
  \url{http://www.usna.edu/cs/~qepcad/}.} These examples are described
in Appendix \ref{appendix:examples} and available as a plain text file
openable with Maple at \url{http://www-polsys.lip6.fr/~greuet/}.

{\small\begin{center}
  \begin{tabular}{|ccccccccc|} \hline
   & $n$ & $d$ & $D$ & $s$ & obj terms & terms & time & QEPCAD\\
    \hline
   nonreached & 3 & 4 & 1 & 1 & 4 & 1 & 2.3 sec. & 0.03 sec.\\
   nonreached2 & 3 & 10 & 3 & 1 & 5 & 5 & 2 sec. & $\infty$\\
   isolated & 2 & 4 & 3 & 1 & 2 & 2 & 0.8 sec. & 0.04 sec.\\
   reachedasymp & 3 & 14 & 1 & 1 & 3 & 1 & 1 sec. & 7.3 sec.\\
   GGSZ2012 & 2 & 2 & 3 & 1 & 2 & 2 & 0.6 sec. & 10.5 sec.\\
   Nie2010 & 3 & 6 & 1 & 1 & 7 & 4 & 1.3 sec. & $\infty$\\
   LaxLax & 4 & 4 & 1 & 3 & 5 & 2 & 0.6 sec. & $\infty$\\
   maxcut5-1 & 5 & 2 & 2 & 5 & 11 & 2 & 0.3 sec. & $\infty$\\
   maxcut5-2 & 5 & 2 & 2 & 5 & 11 & 2 & 0.3 sec. & $\infty$\\
   Coleman5 & 8 & 2 & 2 & 4 & 8 & 4 & 5 sec. & $\infty$\\
   Coleman6 & 10 & 2 & 2 & 5 & 10 & 4 & 33 sec. & $\infty$ \\
   Vor1 & 6 & 8 & n/a & 0 & 63 & n/a & 2 min. & $\infty$\\
    \hline
  \end{tabular}
\end{center}}

\appendix
\section{Description of examples}\label{appendix:examples}
\begin{example}{nonreached, nonreached2}
  Let $g\left(x_1,x_2,x_3\right) = x_1^2-x_1x_2+x_1x_2x_3+x_2+3$ and
  consider the two problems
  \[\left\{\begin{array}{rl}
    \inf_{x\in \mathbb{R}^3} & (x_1x_2-1)^2+x_2^2 + x_3^2 + 42\\
    \text{s.t. } & \ x_3 = 0. \\
  \end{array}\right.\]
 \[\left\{\begin{array}{rl}
    \inf_{x\in \mathbb{R}^3} & (x_1x_2-1)^2+x_2^2 + x_3^2g + \left(x_1+1\right)g^3 + 42\\
    \text{s.t. } & \ g\left(x_1,x_2,x_3\right) = 0. \\
  \end{array}\right.\]
Their infima are not reached. They are the limit of sequences
$f(z_k)$, where $\left\|z_k\right\|$ tends to infinity. For instance,
$z_k$ can be of the form $\left(x^{(k)}_1,\frac{1}{x^{(k)}_1},x^{(k)}_3\right)$, where
$x^{(k)}_1$ tends to infinity. Note that both examples cause instabilities
to numerical algorithms.
\end{example}

\begin{example}{isolated}
  It is a toy example: $f^\star$ is isolated in $f\left(V\cap\R^n\right)$.
  \[\left\{\begin{array}{rl}
    \inf_{x\in \mathbb{R}^2} & \left(x_1^2+x_2^2-2\right)\left(x_1^2+x_2^2\right)\\
    \text{s.t. } & \ \left(x_1^2+x_2^2-1\right)\left(x_1-3\right) = 0. \\
  \end{array}\right.\]
On $V\cap\R^n$, either $x_1^2+x_2^2 = 1$ or $x_1 = 3$. Then the
objective polynomial is either equal to $-1$ or
$\left(7+x_2^2\right)\left(9+x_2^2\right)$. The second expression is
positive over the reals.
\end{example}

\begin{example}{reachedasympt}
  The infimum is both attained and an asymptotic value. Indeed,
  $f^\star = 42$ is reached at any point $\left(x_1,0,0\right)$, but
  is also the limit of sequences of the form
  $\left(x_1,\frac{1}{x_1},0\right)$ when $x_1$ tends to
  infinity. Some iterative methods do not return a minimizer close to
  $\left(x_1,0,0\right)$.
  \[\left\{\begin{array}{rl}
    \inf_{x\in \mathbb{R}^3} & \left(10000\left(x_1x_2-1\right)^4 + x_1^6\right)x_2^6 + \frac{1}{124}x_3^2 + 42\\
    \text{s.t. } & \ x_3 = 0. \\
  \end{array}\right.\]
\end{example}

\begin{example}{GGSZ2012}
  It comes from \cite{GrGuSaZh} (Example 4.4). The minimizer does not
  satisfy the KKT conditions.
  \[\left\{\begin{array}{rl}
      \inf_{x\in \mathbb{R}^2} & (x_1+1)^2+x_2^2\\
      \text{s.t. } & \ x_1^3 = x_2^2. \\
  \end{array}\right.\]
\end{example}

\begin{example}{Nie2011}
  It comes from \cite{NieExact2011} (Example 5.2). It has been studied
  in \cite{GrGuSaZh} because of the numerical instabilities that
  occurs with numerical algorithms.
  \[\left\{\begin{array}{rl}
  \inf_{x\in \mathbb{R}^3} & x_1^6+x_2^6+x_3^6+3x_1^2x_2^2x_3^2-x_1^2(x_2^4+x_3^4)-x_2^2(x_3^4+x_1^4)-x_3^2(x_1^4+x_2^4)\\
  \text{s.t. } & \ x_1+x_2+x_3-1 = 0. \\
  \end{array}\right.\]
\end{example}

\begin{example}{LaxLax}
  The objective polynomial appears in \cite{laxlax} and
  \cite{KLYZ12}. Its infimum is $0$ and is reached over $V\left(x_1 ,
    x_2 - x_3, x_3 - x_4\right)\cap\R^n$.
  \[\left\{\begin{array}{rl}
      \inf_{\left(x\right)\in \mathbb{R}^4} & x_1x_2x_3x_4 - x_1\left(x_2 - x_1\right)\left(x_3 - x_1\right)\left(x_4 - x_1\right)\\
      & - x_2\left(x_1 - x_2\right)\left(x_3 - x_2\right)\left(x_4 - x_2\right) - x_3\left(x_1 - x_3\left)\right
        (x_2 - x_3\right)\left(x_4 - x_3\right)\\
      & - x_4\left(x_1 - x_4\right)\left(x_2 - x_4\right)\left(x_3 - x_4\right)\\
      \text{s.t. } & x_1 =  x_2 - x_3 =  x_3 - x_4 = 0.
  \end{array}\right.\]
\end{example}

\begin{example}{maxcut5-1/5-2}
  A cut of a graph with weighted edges is a partition of the vertices
  into two disjoint subsets. Its weight is the sum of the weights of
  the edges crossing the cut. The maxcut problem is to find a cut
  whose weight is greater than or equal to any other cut.  This
  problem has applications, among other, in very-arge-scale
  integration circuit design and statistical physics \cite{DeLa,
    FePaReRi}. It can be reformulated has a constrained polynomial
  optimization problem \cite{commander}. For a graph of $p$ vertices
  and weight $w_{ij}$ for the edge joining the $i$-th vertex to the
  $j$-th one, it is equivalent to solve
  \[\left\{\begin{array}{rl}
      \inf_{x\in \mathbb{R}^p} & \,\,\,\, - \frac{1}{2}\sum_{1\leq i < j \leq p} w_{ij}\left( 1 - x_ix_j\right)\\
      \text{s.t. } & x_i^2 - 1 = 0, \text{ for } i\in\left\{1,\dots, p\right\}, \\
    \end{array}\right.\]
  We use the set of weight $W_{G5-1}$ and $W_{G5-2}$ in \cite{BaBu},
  that leads to solve
  \[\left\{\begin{array}{rl}
      \inf_{x\in \mathbb{R}^5} & -98 + \frac{23}{2} x_1x_2 + 8 x_1x_3 + 9 x_1x_4 + \frac{17}{2}x_1x_5 + \frac{25}{2} x_2x_3 \\
      &  + 13 x_2 x_4 +\frac{23}{2} x_2 x_5 + 7x_3x_4 + 12 x_3x_5 + 5x_4x_5\\
      \text{s.t. } & x_i^2 - 1 = 0, \text{ for } i\in\left\{1,\dots, 5\right\}.\\
    \end{array}\right.\]
  and
  \[\left\{\begin{array}{rl}
    \inf_{x\in \mathbb{R}^5} &-31+3x_1x_2+3x_1x_3+4x_1x_4+5x_1y_5+\frac52x_2x_3+\frac52x_2x_4+3x_2x_5\\
    & +2x_3 x_4+3x_3x_5+3x_4x_5\\
    \text{s.t. } & x_i^2 - 1 = 0, \text{ for } i\in\left\{1,\dots, 5\right\}.\\
  \end{array}\right.\]
\end{example}

\begin{example}{coleman5/6}
  They come from optimal control problems and appears in
  \cite{coleman}.  For $M\in\left\{5,6\right\}$, let
  $x_1,\dots,x_{M-1}$ and $y_1,\dots,y_{M-1}$ be the indeterminates.
  \[\left\{\begin{array}{rl}
      \inf_{\left(x,y\right)\in \mathbb{R}^{2M}} & \,\,\,\, 
      \frac{1}{M}\sum_{i = 1}^{M-1} x_i^2 + y_i^2\\
      \text{s.t. } & y_1 - 1 =  y_{i+1} - y_i - \frac{1}{M-1}\left(y_i^2 - x_i\right) = 0, \text{ for } i\in\left\{1,\dots, M-2\right\}.
  \end{array}\right.\]
\end{example}

\begin{example}{Vor1}
  It comes from \cite{Vor} and have no constraints. It is too large to
  be written here but can be found at
  \url{http://www-polsys.lip6.fr/~greuet/}.
\end{example}


\bibliographystyle{siam} \bibliography{biblio}
\end{document}